\documentclass[aps,prb,twocolumn,groupedaddress,floatfix,showpacs]{revtex4-1}

\usepackage{epsfig}
\usepackage{amssymb}
\usepackage{amsmath}
\usepackage{bm}
\usepackage{hyperref}
\usepackage{enumerate}

\newtheorem{theorem}{Theorem}[section]

\newtheorem{corollary}[theorem]{Corollary}

\newenvironment{proof}[1][Proof]{\begin{trivlist}
\item[\hskip \labelsep {\bfseries #1}]}{\end{trivlist}}
\newenvironment{definition}[1][Definition]{\begin{trivlist}
\item[\hskip \labelsep {\bfseries #1}]}{\end{trivlist}}

\newcommand{\qed}{\nobreak \ifvmode \relax \else
      \ifdim\lastskip<1.5em \hskip-\lastskip
      \hskip1.5em plus0em minus0.5em \fi \nobreak
      \vrule height0.75em width0.5em depth0.25em\fi}
\begin{document}

\title{Multiorbital effects in the functional renormalization group:\\ A weak-coupling study of the Emery model}
\author{Stefan A. Maier$^1$\thanks{Electronic address: {\tt smaier@physik.rwth-aachen.de}} %
Jutta Ortloff$^{1,2}$, and Carsten Honerkamp$^1$}
\affiliation{$^1$ Institute for Theoretical Solid State Physics, RWTH Aachen University, D-52074 Aachen, Germany 
\\ and JARA - FIT Fundamentals of Future Information Technology \\
$^2$ Theoretical Physics, University of W\"{u}rzburg, D-97074 W\"{u}rzburg, Germany\\
}

\date{\today}

\begin{abstract}
We perform an instability analysis of the Emery three-band model at hole doping and
weak coupling within a channel-decomposed functional renormalization
group flow proposed in Phys.\ Rev.\ B {\bf 79}, 195125 (2009). In our approach, momentum dependences are taken into account with improved precision compared to previous studies of related models. Around a generic
parameter set, we find a strong competition of antiferromagnetic and
$d$-wave Cooper instabilities with a smooth behavior under a variation of
doping and additional hopping parameters. For increasingly incommensurate ordering
tendencies in the magnetic channel, the $d$-wave 
pairing gap is deformed at its maxima.
 Comparing our results for the Emery model to
those obtained for the two-dimensional one-band Hubbard model with
effective parameters, we find that, despite considerable qualitative
agreement, multi-orbital effects have a significant impact on a
quantitative level.
\end{abstract}


\maketitle

\section{Introduction}

Unconventional superconductivity is an extensively discussed topic of condensed matter physics.
For cuprate high-$T_{\rm c} $ materials, resonating valence-bond\cite{RVB} and spin-fluctuation\cite{scalapino-review} mechanisms have been proposed besides other approaches that require low-energy bosons in addition to an electron-electron interaction. While the former relies on a strong-coupling scenario, the latter may also apply to iron-based and other types of unconventional superconductors where the electrons interact more weakly.

A particular problem in the theoretical description of spin-fluctuation induced superconductivity is that often there is no  clear separation of energy scales, i.e.\ the spin-fluctuations are built up at least in parts by the same  electronic degrees of freedom as the pairing tendencies. At weak coupling, renormalization group (RG) methods can be used to study both types of electronic fluctuations on equal footing.  Both,
for the two-dimensional one-band Hubbard model at weak coupling\cite{schulz,lederer,furukawa} and for iron-based superconductors,\cite{chubukov} the interplay of the spin-density wave (SDW) and superconducting (SC) channels has been investigated
within RG approaches involving a small number of running couplings. Such studies have helped to understand the interacting ground states qualitatively.
More generally, the RG interpolates between microscopic models at high scales and and effective models at lower scales. As they take into account the dominant fluctuations on equal footing, RG methods may provide a many-body framework that can become quantitatively precise at least for some unconventional superconductors whose parameters fall into the weak-coupling sector.

In the attempt to become more quantitative, functional renormalization group (fRG) techniques for fermions (for a recent review, see Ref.~\onlinecite{Metzner-Rmp}) shall be useful, as they allow for studying the flow of momentum-dependent coupling functions around the Fermi surfaces which can resolve more details of the models under investigation. Using fRG, the ordering tendencies of the one-band Hubbard model at weak coupling have been classified in a number of works (see, e.g.\ Refs.~\onlinecite{halboth-metzner,n-patch,rohe}). Among others, also models for iron-based superconductors have been investigated within the fRG framework.\cite{wang-2009,wang-2010,platt-first} Here, relevant variations of the superconducting gap function due to material-dependent electrons structure differences have been found.\cite{pnictides-gap} 

Generically, the  multiple Fermi surfaces (hole and electron pockets) of iron-based superconductors with their varying orbital character have been proven to add new aspects such as competing pairing channels and gap anisotropies.\cite{pnictides-gap,sidplatt} Triggered by these insights, multiband effects are now explored and revealed in a number of other systems (e.g. Refs. \onlinecite{uebelacker,kiesel1,kiesel2}). 
It therefore appears rewarding to study multiband effects
in models for the cuprates. The main question is whether the underlying multiband character gives rise to deviations from the one-band picture, even though the single band at the Fermi level has overwhelming $d_{x^2-y^2}$-character. Furthermore, some proposed ordering phenomena in Copper oxide planes such as ring currents\cite{varma,fischer} are not describable by models involving one $d$-orbital on the Copper atoms only. Hence, also from this perspective, an fRG study of the Emery model,\cite{Emery,andersen-YBCO}  which includes the oxygen $p$-orbitals, appears to be worthwhile.

 Of course, within an fRG framework, cuprate models have to be considered at (possibly unrealistic) weak coupling strengths. However, as for the one-band Hubbard model, similarities between weak and strong coupling behavior can be expected. Moreover, multi-orbital cuprate models
seem very well suited as a testbed for methodological developments.
For example, the impact of an interaction rendered nonlocal by the transformation to the band language\cite{tmaier} 
can be more easily studied than in the context of pnictide models. A recent fRG study of two- and three-orbital models for the cuprates with orbitals only residing on the Copper atoms has revealed that these so-called \emph{orbital-makeup} effects may have a significant impact on the phase diagram.\cite{uebelacker}

In the strong-coupling case, other powerful many-body methods are applicable, such as dynamical mean-field theory\cite{Kotliar-review-DMFT,Vollhardt-review-DMFT,weber-2008,Medici-emery-model,weber-2010a,weber-2010b,wang-tpp,weber-charge-trans} (DMFT), dynamical cluster quantum Monte Carlo (DCQMC) techniques\cite{Hettler-DCA,Kent-emery-model} and the so-called variational cluster approach\cite{Potthof-VCA,wuerzburg-graz-VCA,VCA-twogap,Hanke-Kiesel-VCA} (VCA). Also these methods are still being extended and developed further in order to hopefully promote a deeper understanding of the cuprate superconductors. 
Unfortunately, a practicable strong-coupling truncation of the fRG flow equations for fermions, which would allow for direct comparison between fRG and those other methods, is not known so far.
As the fRG approach pursued in this work is hence confined to weak coupling, we do not primarily seek to compare the results of these strong-coupling approaches to our own findings for weak coupling. 

In addition to a weak-coupling truncation which neglects three-particle and higher interaction terms, the fRG studies mentioned so far share a common feature: While the dependence of the interaction on the Matsubara frequencies of the fermionic fields is omitted by projecting to zero frequency, their momentum dependence is projected to the Fermi surface, which is divided into $ n$ patches. In order to keep the numerics tractable at increasing resolution, a channel decomposition has been proposed for the frequency dependence\cite{karrasch} and for the momentum dependence\cite{husemann_2009} of the coupling functions. In such an approach, the coupling function of the two-particle interaction is split into at least three contributions which depend strongly only on one frequency/momentum variable and weakly on the other two ones.
For the one-band Hubbard model in two dimensions, the projection to zero frequency has been relaxed in a channel-decomposed approach.\cite{HGS-freq-dep,giering-new} In all those works, the weak momentum dependences have been accounted for by simple ans\"atze within an \emph{exchange parametrization}. These ans\"atze have been found to describe the weak momentum dependences well in large parts of the phase diagram with the exception of a region around the transition from $d$-wave SC to ferromagnetism (FM).\cite{jutta} Let us note in passing that also graphene and the kagome lattice models have recently been studied within a related channel-decomposed fRG treatment, dubbed singular-mode-fRG (SMFRG).\cite{ch-dec-graphene,ch-dec-kagome}

As for fRG studies of models for Copper oxide planes, the phase diagrams in Ref.~\onlinecite{uebelacker} apparently show some discretization artifacts related to the projection onto the Fermi surface. Hence, employing a channel-decomposed fRG approach for multi-orbital models appears natural:  Since the orbital makeup induces some non-locality in the interaction of multiband Hubbard models, it would be advantageous to resolve its momentum dependence away from the Fermi surface. Moreover, the material characteristics contained in such a multi-orbital model should appear in the high- rather than in the low-energy sector of these models, which makes some momentum resolution perpendicular to the Fermi surface even more desirable. Furthermore, the ans\"atze for the weak momentum dependences underlying the 
exchange-parametrization of Ref.~\onlinecite{husemann_2009} seem questionable in the presence of orbital makeup.
Therefore, instead of restricting the study to a few form factors, we will use the channel-decomposed flow with momentum dependences discretized on chess-boards in the present work. Note that this can be viewed as wavevector-based band-picture analog of the representation of the channels in terms of fermion bilinears on (short) bonds of the lattice in the real-space/orbital picture that is at the heart of the SMFRG approach by Q.-H. Wang \emph{et al.} \cite{ch-dec-graphene,ch-dec-kagome}. As we are primarily interested in a comparison to previous work in the band picture, the wavevector-based approach seems more adequate, but in principle the two  setups can be transformed into each other. The SMFRG approach usually deals with all orbitals or bands of the effective model together, while our current study is reduced to the conduction band at the Fermi surface.

This work is dedicated to a channel-decomposed fRG instability analysis of the Emery model and to the comparison to the one-band Hubbard model with effective parameters. It is organized as follows: In Sec.~\ref{sec:model}, we introduce the Emery model and the corresponding effective action for its conduction band. Then we give a prescription how the parameters of an effective one-band Hubbard model can be calculated and classify different kinds of multi-orbital effects. As we show in Appendix~\ref{sec:eff-pointgroup}, our coupling functions of the Emery model show the same
trivial point-group behavior as those of the one-band Hubbard model and therefore this prescription is viable. Sec.~\ref{sec:fRG} is devoted to the fRG formalism. After introducing the general form of the RG flow equations, we elaborate on a channel decomposition of those flow equations in Sec.~\ref{sec:chan-dec}. In Sec.~\ref{sec:FFE}, we further comment on the exchange parametrization of Refs.~\onlinecite{husemann_2009,HGS-freq-dep,giering-new} and on how approximations made in these studies can be understood in a group-theoretic sense.

Our numerical results for the Emery model are discussed in Sec.~\ref{sec:results}. First, we comment on the strong competition of the antiferromagnetism (AFM) and $d$SC instabilities observed for most parameter sets considered in this work in Sec.~\ref{sec:leading_inst}. The subsequent discussion of the dependence of the stopping scale on doping and oxygen-oxygen hopping suggests that the system exhibits a first-order transition between these to phases. In Sec.~\ref{sec:eff-short-range}, we compare the phase diagrams of the Emery model and of the one-band Hubbard model with the corresponding parameters. We finally conclude in Sec.~\ref{sec:conclusion} by discussing the importance of different kinds of orbital makeup effects.

\section{Model} \label{sec:model}
\subsection{Three-orbital Emery model}
 In this work, we study a three-orbital model introduced by Emery\cite{Emery,andersen-YBCO} for the description of the Cu-O planes of the high-$T_{\rm c} $ compounds. Its Hamiltonian reads as
\begin{equation*}
  H  = \sum_{{\bf k},\sigma} \Psi_\sigma^\dagger ({\bf k}) \left( {\cal H}_0  - \mu \right) \Psi_\sigma ({\bf k}) \, + : H_\mathrm{int} : \, ,
\end{equation*}
where $ {\bf k} $ denotes the 2D-momentum, $ \mu $ the chemical potential and $ \sigma $ the spin index
and where the interaction term $H_\mathrm{int} $ enters in normally ordered form.
The annihilation operators $ d_\sigma ({\bf k}) $ for the Cu $3d$-orbitals and $ p_{x,\sigma} ({\bf k}) , p_{y,\sigma} ({\bf k})  $ for the oxygen $2p$-orbitals form the components of the orbital pseudo-spinor
\begin{equation*}
\Psi_\sigma ({\bf k}) = \left( \begin{array}{c} d_\sigma ({\bf k}) \\ p_{x,\sigma} ({\bf k}) \\ p_{y,\sigma} ({\bf k}) \end{array}  \right) \, .
\end{equation*} 
Here and throughout we set the lattice constant (i.e.\ the spacing between neighboring Copper atoms) to unity. The one-particle part of $H$ is then determined by the matrix
\begin{align} \label{eqn:1pHam}
 {\cal H}_0 & = \left( \begin{array}{ccc}  \epsilon_d &  t_{pd} \, s_x &  t_{pd} \,s_y \\ 
 t_{pd} \, s_x& \epsilon_p+ t_{pp} \, c_x& 2 t_{pp}\, s_x  s_y\\
 t_{pd} \,s_y & 2 t_{pp} \, s_x s_y& \epsilon_p+ t_{pp} \, c_y\end{array} \right)  \, , \\ \notag
  s_{x,y} & = \sin \left(k_{x,y}/2 \right) \, , \quad c_{x,y} = \cos \left(k_{x,y} \right) \, ,
\end{align}
with on-site energies $ \epsilon_d, \epsilon_p $ and hopping integrals $ t_{pp}, t_{pd} $. 
We use LDA-values of these parameters\cite{hybertsen} for ${\mathrm{La}}_{2}$${\mathrm{CuO}}_{4}$ as a starting point of our analysis.
The importance of the oxygen-oxygen hopping $t_{pp} $ has been extensively discussed\cite{andersen-YBCO,weber-2008,weber-2010a,weber-2010b,wang-tpp}
 for the strong coupling case. At weak coupling, in its absence the leading instability of the system would correspond to commensurate AFM due to perfect nesting.
In the basis chosen in Eq.~(\ref{eqn:1pHam}), $ {\cal H}_0 $ is not $ 2 \pi $-periodic due to the $s_{x/y}$-entries. As pointed out in Ref.~\onlinecite{ours-point-group}, all momenta still must be folded back to the first Brillouin zone.

In addition to this one-particle Hamiltonian, we consider a short-ranged interaction term
\begin{align*}
 H_\mathrm{int} & =   U_d \sum_i  n_{d,\uparrow} ({\bf r}_i) \, n_{d,\downarrow} ({\bf r}_i) + U_p \sum_j   n_{p,\uparrow} ({\bf r}_j) \, n_{p,\downarrow} ({\bf r}_j)  \\
  & \quad + U_{pd} \sum_{\langle i j \rangle}  n_d ({\bf r}_i) \, n_p ({\bf r}_j) + U_{pp} \sum_{\langle j j' \rangle}   n_p ({\bf r}_j) \, n_p ({\bf r}_{j'})  \, ,
\end{align*}
where the brackets $ \langle i j \rangle $ and $ \langle j j' \rangle $  indicate that the sum only runs over neighboring orbitals of the respective types.
We will restrict our study to weak interaction, i.e.\ the typical energy scales of the interaction are
 about one order of magnitude below the values given in Ref.~\onlinecite{hybertsen} for ${\mathrm{La}}_{2}$${\mathrm{CuO}}_{4}$. Interaction terms involving the oxygen $p$-orbitals are weak compared the the dominating $U_d$-term and are thus often neglected in the literature (see Ref.~\onlinecite{wang-tpp} for example).

Within the fermionic fRG approach pursued in this work, it turns out to be advantageous to write the quadratic part of the Hamiltonian in its diagonalized form, i.e.\ in terms of bands instead of orbitals. The field operators then correspond to Bloch states that do not get mixed by the one-particle part of the Hamiltonian. The band dispersion of the Emery model is depicted in Fig.~\ref{fig:bs} for typical parameter values. The chemical potential $ \mu $ is chosen to values around van-Hove filling $ \mu_{\rm vH} $ where the Fermi surface touches the saddle points at $ (0, \pi ) $ and $ (\pi,0) $ of the uppermost band. We then obtain one conduction band which is separated
from two valence bands by an energy gap of about four times its width.
   Through the orbital weight imposed by the unitary transformation from the orbital to the band picture,
the interaction acquires a nontrivial momentum dependence, dubbed \emph{orbital makeup}\cite{tmaier} by some authors.
From the form of the one-particle Hamiltonian Eq.~(\ref{eqn:1pHam}), one finds that the hybridization of the $d$- and $p$-orbitals grows from the center to the boundary of the Brillouin zone.
\begin{figure}
 \includegraphics[width=8.4cm]{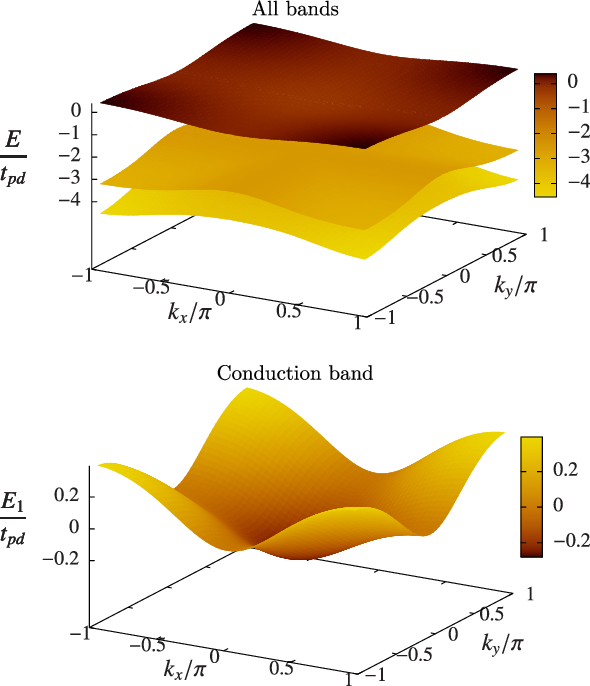}
 \caption{(Color online) Bandstructure (upper part) of the Emery model for the dispersion parameters given in Tab.~\ref{tab:params} and at van-Hove filling and a blow-up for the conduction band (lower part).}
\label{fig:bs}
\end{figure}
\subsection{Effective one-band action} \label{sec:eff-1bd}
 In principle, an appropriate  low-energy solver could be directly applied to the full three-band model. Such a solver would effectively resum diagrams up to infinite order, even if the underlying concept is of non-perturbative nature.  The result of such a resummation at weak to moderate coupling will be dominated by diagrams with internal legs on the conduction band. It should therefore suffice to treat the valence bands perturbatively only up to a certain order. 

In the fRG literature on multi-orbital models (see for example Ref.~\onlinecite{Thomale-Platt-pnictides}) high-energy modes above some ultraviolet cutoff or in bands that do not cross the Fermi level are usually neglected. 
In a recent publication,\cite{3particle} we have discussed
the impact of the most relevant perturbative corrections from the modes above this cutoff. In the present work, those modes correspond to the valence bands. As a starting point, we have considered an effective action
\begin{align} \notag
 S_{\rm eff} [\bar{\chi}_-,\chi_-] 
& = \boldsymbol{\bar{\chi}}_- \boldsymbol{D}_- \boldsymbol{\chi}_- + {\cal V} [\bar{\chi_-},\chi_-] \, ,  \\ \label{eqn:Seff}
e^{-{\cal V} [\bar{\chi_-},\chi_-]} & =  \int \!\! {\cal D} \chi_+ \, e^{-\boldsymbol{\bar{\chi}}_+ \boldsymbol{D}_+ \boldsymbol{\chi}_+}
 e^{-S^{(4)} [\bar{\chi}_+ + \bar{\chi}_-, \chi_+ + \chi_-]} \, ,
\end{align}
 for the conduction band represented by the Grassmann field $ \chi_- $ with the inverse propagator $ \boldsymbol{D}_- $ as a starting point.
 Here, the functional integral defining the effective interaction
$ {\cal V} $ runs over the Grassmann field $ \chi_+ $ representing the valence bands with inverse propagator $ \boldsymbol{D}_+ $ and $ S^{(4)} $ denotes the bare interaction of 
the original model. Formally, all connected Greens functions of the conduction band can be reproduced exactly from $ S_{\rm eff} $. In practice, however, the functional integral in Eq.~(\ref{eqn:Seff}) is evaluated perturbatively. In Ref.~\onlinecite{3particle}, we have argued that the most relevant corrections to the bare interaction of the conduction electrons (or holes)
are given by a second-order tree-diagram  with six external legs on the conduction band. This diagram is depicted in the second line of Fig.~\ref{fig:6pt}.
 In the following, we will also study the impact of the three-particle interaction on the fRG flow.
\begin{figure}
 \includegraphics[width=8.4cm]{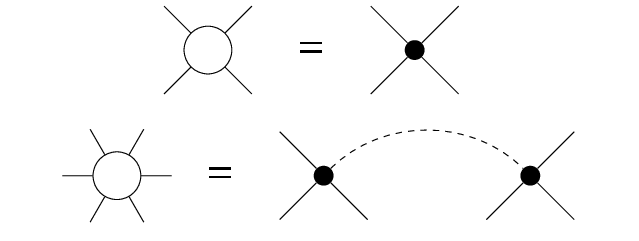}
 \caption{Two- and three-particle vertex of $ S_{\rm eff} $ for the effective one-band model. Small filled vertices correspond to the bare interaction $ S^{(4)} $, whereas the vertices of $ S_{\rm eff} $ on the left hand sides are denoted by empty circles. Solid lines represent the conduction band and dashed lines the valence bands. There are no propagators attached to the external legs. Self-energy effects will be neglected.}
 \label{fig:6pt}
\end{figure}

 If this three-particle term is neglected, the $(t,t')$ one-band Hubbard model
seems to be a good candidate for an effective low-energy model.
 Such a description can be considered valid if a low-energy solver yields similar results for $ S_{\rm eff} $  and for the action corresponding to the one-band \emph{Hamiltonian}
\begin{align} \notag
 H = & \left[ t \sum_{\langle i j \rangle,\sigma}  c^\dagger_{\sigma,i} c_{\sigma,j} + t' \! \sum_{\langle \langle i j \rangle \rangle,\sigma} c^\dagger_{\sigma,i} c_{\sigma,j} + \text{h.c.} \right] \\ \label{eqn:hubbard} & + U_{\rm eff} \sum_i n_{\uparrow,i} \, n_{\downarrow,i} \, ,
\end{align}
 where $ n_{\sigma,i} = c_{\sigma,i}^\dagger c_{\sigma,i} $ and where $ \langle \langle i j \rangle \rangle $ indicates that the sum runs over second neighbors.

 Before we give a prescription how to calculate the effective parameters $t,t'$ and $ U_{\rm eff}$, let us elaborate
on the differences between the effective action Eq.~(\ref{eqn:Seff}) for the conduction band and the effective one-band Hamiltonian Eq.~(\ref{eqn:hubbard}).
 The frequency dependence of the action corresponding to Eq.~(\ref{eqn:hubbard}) will of course be trivial in contrast to the frequency dependence in $ S_{\rm eff} $ which is generated by integrating out the valence
 bands. Throughout this work, we will, however, project to zero frequency and hence we will not discuss such effects. Recently, frequency-dependent RG flows have been analyzed for the two-dimensional one-band Hubbard model.\cite{HGS-freq-dep,giering-new} In this context, also frequency-dependent multi-orbital effects appear to be worth further study.

 As far as the momentum dependence is concerned, the Hamiltonian Eq.~(\ref{eqn:hubbard}) is short-ranged in the sense that all terms are either on-site or describe hopping between first and second neighbors at most. In contrast, $ S_{\rm eff} $ may contain quite long-ranged terms both in the dispersion and in the two-particle interaction. At van-Hove filling, the restriction to only two
hopping integrals in Eq.~(\ref{eqn:hubbard}) can be justified in the spirit of a gradient expansion. Since we have a diverging density of states at the saddle points $ A =(0,\pi) $ and $ B= (\pi,0) $, the integrals over momenta of internal lines in a perturbation expansion will be dominated by a small region around these saddle points. Consequently, in leading order in an expansion around the van-Hove points, the exact and an approximate dispersion should coincide. Since a $ k_x k_y $-term is forbidden by symmetries, only two expansion coefficients remain in leading (second) order. They can be expressed in terms of effective first and second neighbor hoppings $ t $ and $ t'$. 

Away from van-Hove filling, the situation may be more involved and hopping terms between non-neighboring unit cells of the direct lattice may be needed for an effective model.
Since this corresponds to Wannier functions that have support on more than one unit cell, such a description is not really of tight-binding type.
The effective two-particle interaction will also bear traces of the multi-orbital character of the underlying model. More precisely, the orbital makeup renders the interaction nonlocal. Whether this feature plays a role for the low-energy physics remains a question to be answered by applying a low-energy solver.

 Before doing so, we should however give a prescription how to obtain the effective Hubbard parameters $ t $, $ t'$ and $ U_{\rm eff} $. From the comments made above, a gradient expansion around the saddle points of 
the dispersion appears natural as a guiding principle. So the calculation of $ t $ and $ t' $ is straightforward. 

For the interaction, we proceed as follows.
First, let $ U (k_1,k_2,k_3) $ denote the coupling function that appears in the parametrization of the effective two-particle interaction
\begin{align} \notag
 {\cal V}^{(4)}  [\bar{\chi}_-,\chi_-] &= - \frac{1}{4} \int \prod_i d (\sigma_i,k_i) \,  \bar{\chi}_{-,\sigma_1} (k_1) \, \bar{\chi}_{-,\sigma_2} (k_2) \\ \notag
 & \quad \times \chi_{-,\sigma_3}  (k_3) \, \chi_{-,\sigma_4} (k_4) \, \delta (k_1 +k_2 -k_3 -k_4)\\ \notag
 & \quad \times  \left[ U (k_1,k_2,k_3) \, \delta_{\sigma_1, \sigma_4}
 \delta_{\sigma_2, \sigma_3} \right. \\ \label{eqn:parameter}
 & \quad \left. - U (k_2,k_1,k_3) \, \delta_{\sigma_1, \sigma_3} \delta_{\sigma_2, \sigma_4} \right] 
\end{align}
which exploits the U(1), SU(2) and translation symmetries.\cite{salm_hon_2001} We further assume a basis of Bloch states
that ensures the $ C_{4v} $-symmetry of this coupling function.
In Appendix~\ref{sec:eff-pointgroup}, we will show that such a Bloch basis exists. This property is nontrivial, since the oxygen $p$-orbitals are mapped onto one another upon a rotation of $ \pi/2$.

As a second step, we also expand the two-particle interaction around the saddle-points. In leading order, we then have four two-particle 
couplings%
, namely
\begin{align*}
 g_1 & = U (A,B,B) = U(B,A,A) \\
 g_2 & = U (A,B,A) = U(B,A,B) \\
 g_3 & = U (A,A,B) = U(B,B,A) \\
 g_4 & = U (A,A,A) = U(B,B,B) \, .
\end{align*}
These are the couplings of the so-called two-patch approximation.\cite{furukawa}
For the Hubbard model, all four two-patch couplings $ g_i $ are equal to $ U_{\rm eff} $. For given $ {\cal V}^{(4)} $, we therefore take the average of the two-patch couplings
$ U_{\rm eff} = \sum_i g_i /4 $ as the effective Hubbard interaction, while the hopping parameters $t$ and $t'$ are calculated from a gradient expansion. This means that we choose the parameters of the one-band Hubbard model such that it has a common two-patch approximation
with $ S_{\rm eff} $ and  that we further restrict the hopping terms to neighboring unit cells and the interaction
to an on-site density-density term in the effective one-band Hamiltonian.

Note that, in contrast to the famous work by Zhang and Rice,\cite{zhang-rice} which has been tailored rather for the strong-coupling case, this method is non-perturbative in the hybridization between the orbitals. 
\subsection{Classification of multi-orbital effects}
 Before we discuss our method, a classification of multi-orbital effects seems to be in order.
 Clearly, if we have a full model [in the present case $ S_{\rm eff}$ given in Eq.~(\ref{eqn:Seff})] 
 and an effective one-orbital Hamiltonian [e.g.\ the Hubbard Hamiltonian in Eq.~(\ref{eqn:hubbard})], effects contained in the full, but not in the effective model are of multi-orbital character.
 Of course, the multi-orbital nature cannot be attributed to a certain effect without referring to a prescription according to which the full model is mapped to its effective single-orbital counterpart.
 So, in the following, multi-orbital effects will be classified with respect to the above prescription based on a
 gradient expansion around the saddle points of the dispersion.

In this picture, multi-orbital effects decay into three classes, namely
\begin{enumerate}[i)]
 \item effects resulting from the three-particle interaction, and other higher-order vertices generated by the high-energy bands, \label{enu:3pe}
 \item orbital-makeup effects, in particular a
 detuning of the two-patch couplings $ g_i $, \label{enu:om}
 \item hopping between non-neighboring unit cells.
\end{enumerate}
 The three-particle and higher-order vertices responsible for the effects listed as \ref{enu:3pe}) appear as (perturbative) corrections to the bare interaction in the effective interaction $ {\cal V}$ in Eq.~(\ref{eqn:Seff}). Since these corrections also contain internal loops with valence-band propagators, they are in general frequency-dependent. 
Throughout this work, we will however neglect four-particle and higher-order terms as well as contributions with closed valence-band loops (see Fig.~\ref{fig:6pt} for the remaining diagrams).
Since we then
 include only tree-diagrams of bare vertices in the perturbative expansion of ${\cal V} $ (i.e.\ diagrams that are reducible to bare vertices by cutting one internal line), the two- and three particle interaction terms of $ S_{\rm eff}$ are frequency-independent in this approximation. Hence, the frequency dependence of $ S_{\rm eff} $ is completely neglected in this work. 

In the present work, the three-particle term of $ S_{\rm eff}$ is either neglected or fed back into the flow equation of the two-particle vertex using an extended truncation of the fRG flow equations.\cite{3particle} It will turn out to play a minor role due to the large gap between the conduction band and the valence bands.

 As far as orbital makeup effects are concerned, we thus restrict the discussion to the two-particle interaction although the three-particle term obviously bears signatures of orbital makeup.
 Deviations of the two-particle interaction from the on-site Hubbard type manifest themselves in a detuning of 
the two-patch couplings $g_i$ and in a non-trivial momentum dependence also away from the saddle points of the dispersion. These effects have been listed as \ref{enu:om}) above. This implies that orbital makeup effects can be partly understood with the help of the flow equations in the two-patch approximation.\cite{furukawa}
By looking at the bare one-particle part Eq.~(\ref{eqn:1pHam}) of the Emery Hamiltonian, however, we observe that
the hybridization between the $d$ and $p$ orbitals is stronger at the saddle points than in other parts of the Brillouin zone. For example, at the origin in $k$-space, there is no hybridization at all. Therefore, we have a pronounced momentum dependence which may lead to effects that cannot be captured in the two-patch approximation.

The effective action $ S_{\rm eff} $ and the effective single-orbital Hamiltonian $ H_{\rm eff} $ also differ in their quadratic parts. If the dispersion of the conduction band in the former is expanded around its saddle points, also hopping terms between non-neighboring unit cells appear in the coefficients in subleading orders.
As already mentioned, these longer-range hopping integrals do not fit well into a tight-binding picture, as the corresponding Wannier functions would have long tails. Since the conduction band is predominantly of $d$-orbital character, the hybridization with oxygen $p$-orbitals can be said to create such tails. 
We expect hopping between non-neighboring unit cells to play a minor role for the fRG results at weak coupling.

\section{Functional renormalization group} \label{sec:fRG}
 \subsection{General framework}
  In order to now extract physical properties at low temperatures for the effective one-band action, we have to resort to  low-energy solvers. In the fRG approach used in this work, an exact flow equation
  interpolates between a microscopic action and the generating functional of the one-particle irreducible (1PI) vertices. More precisely, we are dealing with a first-order differential equation describing the renormalizations of this functional that occur when some infrared cutoff $ \lambda $ is lowered.  The exactness of the flow equation allows for non-perturbative truncations,\cite{btw-review,Metzner-Rmp} where fermionic two-particle interaction is decoupled by a Hubbard-Stratonovich transform. At weak to moderate coupling, one may as well resort to a vertex expansion\cite{salm_hon_2001,Metzner-Rmp} in order to avoid some biases of the partially bosonized approach. The purely fermionic vertex-expansion approach is perturbative in the \emph{effective} interaction (and hence of weak-coupling nature) and \emph{resums} diagrams up to infinite order in the \emph{bare}  interaction in this way.

 If the two-particle interaction is parametrized in the same way as in Eq.~(\ref{eqn:parameter}), but with a renormalized coupling function $ V (k_1,k_2,k_3) $ instead of $ U (k_1,k_2,k_3) $,
the right-hand side of the resulting flow equation
\begin{align*}
 \partial_\lambda  V (k_1,k_2,k_3) & = {\cal T}_{\rm pp} (k_1,k_2,k_3)
 +  {\cal T}^{\rm cr}_{\rm ph}  (k_1,k_2,k_3) \\ & \quad + {\cal T}^{\rm d}_{\rm ph} (k_1,k_2,k_3)
\end{align*}
 for the two-particle vertex comprises five diagrams. $ {\cal T}_{\rm pp} $ denotes the particle-particle diagram and $ {\cal T}_{\rm ph}^{\rm cr} $ the crossed particle-hole diagram while $ {\cal T}_{\rm ph}^{\rm d} $ comprises three direct particle-hole diagrams including vertex corrections and screening.  For their precise form, we refer to Appendix~\ref{sec:flow-eq}.
Three-particle-feedback corrections to these flow equations are given in Appendix~\ref{sec:3p-feedback}.

In this article, we employ the $\Omega $-scheme\cite{husemann_2009} regularization, where the propagator $ G (k) $ is multiplied by a smooth frequency regulator according to
\begin{equation} \label{eqn:Omega-reg}
 G (k) \to G (k) \, R_\lambda (k_0) \, , \quad R_\lambda (k_0) = \frac{k_0^2}{k_0^2+\lambda^2} \, .
\end{equation}
This regularization scheme does not suppress a Stoner instability as a momentum-shell cutoff would.

 \subsection{Channel decomposition} \label{sec:chan-dec}

 In many previous works, frequencies have been projected to zero and the flow of the self-energy has been neglected. More recent studies\cite{HGS-freq-dep,giering-new,uebelacker-se}   on one-band models taking into account parts of the self-energy and frequency-dependent vertices show that the flows to strong coupling are not changed in character if these two approximations are made \emph{simultaneously}. They considerable facilitate the loop integrations in the flow of the two-particle interaction and reduce the number of running couplings. Since our focus
rather lies on orbital makeup effects, we will hence resort to these approximations in the following.

Although this simplifies the flow equations to be solved, a direct and unbiased discretization of all external momenta in the remaining flow equation for the coupling function $ V (k_1,k_2,k_3) $  is still too costly from a numerical viewpoint.
In an older approach to handle this problem, the momentum dependence of $ V $ was projected to a finite number of patches on the Fermi surface.\cite{n-patch} This Fermi-surface 
patching was designed to reproduce the low-energy physics properly, but renormalizations away from the Fermi surface are only crudely approximated. In multi-orbital problems, already the bare interactions, expressed in band representation, show a significant wavevector-dependence away from the Fermi surface.
 As a consequence, the Fermi-surface patching in Ref.~\onlinecite{uebelacker} is plagued by discretization artifacts.
 Hence, an approach that is more suited to capture orbital-makeup effects is desirable.
 (Nevertheless, meaningful results for multiband models can by obtained from Fermi-surface patching fRG as in Refs.~\onlinecite{platt-first,pnictides-gap,sidplatt,Thomale-Platt-pnictides,uebelacker,kiesel1,kiesel2,kiesel3}.)

As in Refs.~\onlinecite{husemann_2009,HGS-freq-dep,giering-new,ch-dec-graphene,ch-dec-kagome}, we therefore decompose the flow into three channels
\begin{align*}
 V(k_1,k_2,k_3)  = U(&k_1,k_2,k_3) - \Phi_{\rm SC} (k_1+k_2,k_1,k_3) \\ & + \Phi_{\rm M} (k_3-k_1,k_1,k_2) \\ &
 + \frac{1}{2} \Phi_{\rm M} (k_2-k_3,k_1,k_2) \\ & - \frac{1}{2} \Phi_{\rm K} (k_2-k_3,k_1,k_2) \, ,
\end{align*}
with the bare interaction $ U $ and the pairing, spin and charge coupling functions $ \Phi_{\rm SC} $, $ \Phi_{\rm M} $ and $ \Phi_{\rm K} $.
These coupling functions are generated during the flow according to
\begin{align*}
 \dot{\Phi}_{\rm SC} (k_1+k_2,k_1,k_3) & =  - {\cal T}_{\rm pp} (k_1,k_2,k_3) \\
 \dot{\Phi}_{\rm M} (k_3-k_1,k_1,k_3) & =  {\cal T}_{\rm ph}^{\rm cr} (k_1,k_2,k_3) \\
 \dot{\Phi}_{\rm K} (k_3-k_1,k_1,k_3) & =  - 2 {\cal T}_{\rm ph}^{\rm d} (k_1,k_2,k_3)  \\ & \quad + {\cal T}_{\rm ph}^{\rm cr} (k_1,k_2,k_1+k_2-k_3) \, .
\end{align*}
The precise form of the one-loop terms on the right-hand side of these equations is given in Appendix~\ref{sec:flow-eq}.
 The first argument of the $ \Phi $s corresponds to the total or transfer momentum of the loops in the particle-particle and particle-hole channels, respectively.
 For weak coupling, the coupling functions therefore should depend strongly on their first argument (bosonic momentum) and only weakly on the other two (fermionic) momenta.
 The dependence on the fermionic momenta can be accounted for by an exchange parametrization based on a form-factor expansion (FFE). In Refs.~\onlinecite{husemann_2009,HGS-freq-dep,giering-new}, only $s$- and $d_{x^2-y^2}$-wave form factors have been taken into account. It has however been shown\cite{jutta} that the remainder terms neglected in that work can have a substantial impact on the stopping scale in some parts of the phase diagram
of the one-band 2D Hubbard model.

In the present article, we therefore pursue a different approach. Instead of employing a FFE, we directly patch all three momenta of the $\Phi$s while we project to zero frequency. The fermionic momenta then live on a much coarser grid than the bosonic ones. The finest resolution is only used for the bosonic momenta around potential divergencies and for the internal loop momenta close to the Fermi surface. If not indicated otherwise, we use $ 6 \times 6 $ fermionic quadratic patches and a bosonic resolution of $ 18 \times 18 $ and of $ 126 \times 126 $ patches away from
and close to possible ordering vectors, respectively.
The Pauli principle, point-group and particle-hole symmetries reduce the number of independent couplings further.  
 Note that, close to the transition from $d_{x^2-y^2}$-wave SC to ferromagnetism in the one-band Hubbard model, this resolution would be still
too coarse in the fermionic momenta.\cite{jutta} For the parameters considered in this work, however, we are far away from such a transition and a deformation of the form factors at
low energies should at least qualitatively be captured within our approach. 

Since we project to zero frequency, the Matsubara sums over internal frequencies can be performed analytically as in Ref.~\onlinecite{husemann_2009}. For the $\Omega$-scheme regularization employed in this work, the remaining two-dimensional loop integrals are numerically challenging. They can be efficiently calculated using an adaptive routine.\cite{genz} For different patches, these loop integrals are evaluated in parallel using OpenMP and then stored. When we subsequently assemble the diagrams, we again OpenMP parallelize the calculation for different external momenta. 

We have sent external fields breaking the U(1), SU(2) and/or space group symmetries to zero right from the beginning. The resulting decrease of computational effort, however, comes
with a price. The limit of vanishing external fields is only physical \emph{after} the thermodynamic limit has been performed. Consequently, spontaneous breaking of a symmetry manifests
itself in a flow to strong coupling. We therefore stop the flow when the maximum of the coupling functions reaches $ 7.7 \,  t_{pd} $, and interpret the stopping scale $ \lambda_{\rm c} $
as an estimate for the critical scale of the respective instability.%
\footnote{As for a conventional momentum shell cutoff, the particle-particle susceptibility in the $\Omega$-scheme diverges at a scale $ \lambda_{\rm c} \propto e^{- 1/(U \rho_0)} $ for a
BCS-type interaction and a constant density of states near the Fermi level. Therefore, $ \lambda_{\rm c} $ is proportional to the BCS critical temperature.}
 Note that this stopping condition is rather weak, since $ 7.7 \, t_{pd} $ corresponds to about thirteen times the bandwidth of the conduction band. In comparison, in Ref.~\onlinecite{husemann_2009}, the flow is stopped once the interactions exceed $2.5$ times the bandwidth. For the present work, we have not chosen a stronger stopping condition in order to keep track of the strong competition of AFM and $d$SC instabilities discussed in Sec.~\ref{sec:leading_inst}.
 \begin{table}
 \begin{tabular}{cccccc}
   $ t_{pp} $ & $ \epsilon_p- \epsilon_d $ & $ U_d $ & $ U_p $ & $U_{pd} $ & $U_{pp} $ \\
\hline
\hline
   $ 0.5 \, t_{pd} $ &  $ - 2.77 \, t_{pd} $ & $ 0.385 \, t_{pd} $ & $ U_d/ 8 $ & $ U_d /16 $ & $ 0 $ \\ 
 \hline
 \end{tabular}
 \caption{Generic parameter set with $ t_{pd} > 0 $. The values for the dispersion have been chosen according to Ref.~\onlinecite{hybertsen}, while we have lowered the interaction parameters given in that work by a factor of $1/20$. The parameters of the corresponding one-band Hubbard model are $ t= 9.8 \cdot 10^{-2} \, t_{pd} $, $ t' = - 0.26 \, t $ and $ U_{\rm eff} = 2.73 \, t $. }
 \label{tab:params}
 \end{table}
 
\subsection{Remarks on form-factor expansions} \label{sec:FFE}
Now we comment on the form-factor expansion (FFE) put forward in Ref.~\onlinecite{husemann_2009}. 
 In particular, we shall analyze to what extend a channel-decomposed, renormalized two-fermion interaction can be conveniently expressed as one resulting from a small number of bosonic channels.
 Regarding the classification of such order-parameter fields, we proceed similarly to Vojta \emph{et\ al.}\ in Ref.~\onlinecite{vojta-class}, where different
types of commensurate ordering within the $d_{x^2-y^2}$-wave superconducting phase have been classified according 
to the irreducible representations (IRs) of the point group. The underlying group-theoretical lemmata will be laid out in Appendix~\ref{sec:group-th}.
 For the SMFRG approach, similar considerations have been undertaken.\cite{wang-private}

In the present case,
 the coupling functions $\Phi_\mathrm{SC}, \Phi_\mathrm{M} $ and $ \Phi_\mathrm{K} $ may be decomposed in the spirit of a Hubbard-Stratonovich transform using a set of orthonormal form factors $ f_i $.
 In the Cooper channel, for example, we have
\begin{equation*}
 \Phi_{\rm SC} (l,q,q') = \sum_{i,j} f_i ( {\bf l}/2 -{\bf q} ) \,  f_j ( {\bf l}/2 -{\bf q'} ) \, D_{ij} (l) \, ,
\end{equation*}
with bosonic propagators  $ D_{ij} (l) $. The form factors play the role of fermion-boson vertices, with indices $i $ and $j$ labeling different bosonic species (flavors). In the present from of the exchange parametrization, they are frequency-independent and can therefore be chosen real. They obey the orthonormality relation $ \int \! d {\bf q} \, f_i ({\bf q}) \, f_j({\bf q})  = \delta_{i,j} $. Here and throughout, integrals $ \int \! d{\bf q} $ run over the whole Brillouin zone and a normalization factor has been absorbed into the measure such that $ \int \! d{\bf q} \, 1 = 1$. For given $ \Phi_{\rm SC} $, the matrix elements $ D_{ij} (l) $
are thus uniquely defined.
 Since the two-particle coupling functions are $C_{4v}$-symmetric,
 it appears natural to choose basis functions of the IRs of $ C_{4v} $ as form factors, for example $ f_s ({\bf q}) = 1 $ for $s$-wave , $ f_{p,\pm} ({\bf q})= \sin (q_x) \pm \sin (q_y) $ for  $p$-wave, and $ f_d ({\bf q}) = \cos(q_x) - \cos (q_y) $ for $d_{x^2-y^2}$-wave. (Note that the IR corresponding to a $p$-wave is two-dimensional, while the other ones are one-dimensional.)
\begin{table} 
\begin{tabular}{|l||c|c|c|c|c|}
 \hline
 & $ s$-wave & $p$-wave & $d_{x^2-y^2}$-wave & $ d_{xy} $-wave & $g $-wave \\
\hline
\hline
$ E $ & $ 1 $ & $ 2 $ & $ 1 $ & $1 $ & $1$ \\
\hline
$2 C_4 $ & $ 1 $ & $ 0 $ & $ -1 $ & $-1 $ & $1$ \\
\hline
$ C_2 $ & $ 1 $ & $ -2 $ & $ 1 $ & $ 1 $ & $1$\\
\hline
$ 2I$ & $ 1$ & $ 0 $ & $ 1 $ & $ -1 $ & $-1 $\\
\hline
$ 2I'$ & $ 1$ & $ 0 $ & $ -1 $ & $ 1 $ & $-1 $\\
\hline
\end{tabular}
 \caption{Character table of $ C_{4v} $. The classes $ I $ and $ I' $ correspond to reflections with respect to the $ (0,1), (1,0) $ axes or the $ (1,1), (1,-1) $ axes,
 respectively.}
 \label{tab:transf-prop}
\end{table}

When one projects to zero frequency, the coupling function can be fully recovered by using a complete set of form factors. By (anti)symmetrizing the real Fourier basis functions
on the first Brillouin zone with respect to the $ C_{4v} $ point-group operations, one can easily construct a complete basis set with elements that transform according to the IRs (cf.\ Tab.~\ref{tab:transf-prop}). By equivalence transformations of the IR, these form factors can be rendered well-behaved under $ C_{4v} $ in the sense of Appendix~\ref{sec:group-th}.
The form factors mentioned above are the most slowly varying basis functions of the respective IRs, which corresponds to the formation of exchange bosons from constituents residing on the same site or on neighboring
unit cells.

For the bosonic ordering vectors $ {\bf l} = (0,0) $ and $ (\pi,\pi) $,
the little group ${\cal L}_{\bf l} $ equals the full point group $ C_{4v} $.
 According to Corollary~\ref{th:nomix}, which is proven in Appendix~\ref{sec:group-th},
matrix elements of $ D(l) $ mixing bosons of inequivalent IRs vanish at these momenta. 
Different form factors transforming according to equivalent IRs may however mix. In the following,
we shall refer to this effect as to the admixture of higher harmonics.
In flavor space, the non-vanishing matrix elements of $D_{ij}$ appear in $d \times d$ blocks corresponding to a $d$-dimensional IR. Note that Schur's first lemma\cite{tinkham-groupth} implies that  all these remaining blocks are then a multiple of the unit matrix, if the form factors are well-behaved in the sense of Appendix~\ref{sec:group-th}, where this statement is proven as Corollary~\ref{th:unit-mat}.

At $ {\bf l} = (0,\pi)$ and  $(\pi,0)$, the little group reduces to $ {\cal L}_{\bf l} = C_{2v} $ and 
therefore, again by virtue of Corollary~\ref{th:nomix}, the five IRs
  of $ C_{4v} $ decay into three sets of form factors that do not mix with another. One contains $s$- and $d_{x^2-y^2} $-wave, one $d_{xy} $- and $g$-wave while the third one purely consists of $p$-wave form factors. This $p$-wave set splits into two, each transforming with a one-dimensional IR of $C_{2v}$. Altogether, this corresponds to the four IRs of the little group. For example, the most slowly varying $p$-wave basis functions $ \sin (q_x) $ and $ \sin (q_y) $ then transform with two inequivalent one-dimensional IRs of $ C_{2v}$, which may be referred to as $p_x$- and $p_y$-wave.

For bosonic momenta $ {\bf l} $ on the boundary of the first Brillouin zone, i.e.\ for $ {\cal L}_{\bf l} = C_s $, there are two such sets, one for $s,p_{x/y},d_{x^2-y^2}$- and the other one for $ p_{y/x},d_{xy}$ and $ g$-wave. 
Again, those two sets correspond to the IRs of the little group.
For bosonic momenta that do not lie on any of the symmetry axes, the little group just contains the identity element and all form factors may get mixed.

Let us now assume that the form-factors are well-behaved in the sense of Appendix~\ref{sec:group-th}.
If the mixing between inequivalent IRs of $C_{4v}$ is neglected, the bosonic propagators of the four one-dimensional IRs then inherit the full $ C_{4v} $ symmetry of the coupling function according to Corollary~\ref{th:1dtrans}.
In contrast, the $p$-wave block still transforms with two-dimensional IR matrices.
Let us note in passing that a mixing of different IRs has already been observed in Ref.~\onlinecite{Otsuki} for the RPA pairing susceptibility at incommensurate Copper pair momenta.

  So far, we have only considered a FFE in the Cooper channel. Of course, such an expansion may as well be performed in the other channels, which are then decomposed as
\begin{align*}
 \Phi_{\rm M} (l,q,q') & =  \sum_{i,j} f_i ( {\bf l}/2 +{\bf q} ) \,  f_j ( - {\bf l}/2 + {\bf q'} ) \, M_{ij} (l) \, , \\
 \Phi_{\rm K} (l,q,q') & =  \sum_{i,j} f_i ( {\bf l}/2 + {\bf q} ) \,  f_j ( - {\bf l}/2 +{\bf q'} ) \, K_{ij} (l) \, .
\end{align*}
If one wishes to simplify the RG flow equations by resorting to a FFE, the expansion has to be truncated behind a few terms in order not to exceed available computational resources. This may be conveniently done in the following way.
\begin{enumerate}[i.)]
 \item Neglect the mixing between inequivalent IRs of $C_{4v}$ (or the respective point group 
for other lattice geometries). \label{enu:no-irmix}
 \item Only consider the most slowly varying form factor among equivalent IRs, i.e.\ neglect the admixture of higher harmonics. \label{enu:slow-fs}
\end{enumerate}
(In addition to these approximations, the $p$-, $d_{xy}$, and $g$-wave sectors have not been taken into account in Ref.~\onlinecite{husemann_2009}, as such form factors can be expected to play a minor role for the one-band Hubbard model at van-Hove filling.) If these approximations are adequate, 
the truncated FFE of the RG flow equations should in principle capture
 the important momentum dependences well. 
If, in contrast, the admixture of higher harmonics plays a role, a large number of bosonic channels might be needed.
At least for the one-band Hubbard model, the above approximations seem to be fine in a large region of the
parameter space.\cite{jutta} The question now is, whether
important orbital makeup effects are still captured within a viable truncation of an FFE.

Since the fermionic momenta are directly put on a grid in this work, we are in a position where we can easily keep track of mixing between inequivalent IRs. We expect this mixing to
 play a minor role, if the ordering vectors of leading and subleading instabilities are $ {\bf l} = 0 $ or $ (\pi,\pi) $ or very close. By diagonalization of the
coupling functions as matrices in $ q $ and $ q'$ with $ {\bf l}$ fixed to one of these ordering vectors, optimized form factors can then be attributed to the respective instabilities. These optimized
 form factors will turn out to be close to the most slowly varying ones for most parameters considered in this paper, but in some cases also higher harmonics play a role. A sensible truncation of the FFE then consists in only retaining  the terms corresponding to the most relevant eigenvalues.
Clearly, the optimized form factors are scale-dependent in an fRG flow. In principle, it should be possible to parametrize this scale dependence.\cite{giering-personal}
Let us note in passing that similar effects have already been discussed within a Bethe-Salpether equation approach.\cite{katanin-gap}

For incommensurate antiferromagnetism, however, the potentially non-zero mixing between inequivalent IRs of the point-group symmetry of the lattice may prohibit the calculation
 of an optimized form factor that is defined on the whole Brillouin zone. In such a case, a faithful truncation 
of the FFE would already contain too many terms to be numerically tractable. We will come back to the question of the applicability of a FFE when we discuss our numerical results in the next section.

\section{Numerical results} \label{sec:results}
 We use the parameters displayed in Tab.~\ref{tab:params} as a generic parameter set for which we run the RG flow in the conventional truncation, i.e.\ without the three-particle feedback. We then test the stability of our results against variation of some of these parameters and against three-particle terms.
 We further compare the Emery model to the corresponding one-band Hubbard model. All results presented in this paper are for zero temperature and all parameter sets considered are on the hole-doped side.

 Note that spectral functions are inaccessible in our approach, since we project to zero frequency and neglect the self-energy.
 But even if these approximations were relaxed within an instability analysis, the flow would be stopped at a finite scale, where contributions to the
 interaction start to diverge. Therefore, the spectral functions obtained from fRG and successive analytic
 continuation then would still not really comparable their DMFT counterparts in Refs.~\onlinecite{wang-tpp,Medici-emery-model}.
 \subsection{Nature of the leading instability} \label{sec:leading_inst}
  For the generic parameters and at van-Hove filling, we observe a flow to strong coupling at about $ \lambda_{\rm c} = 2.9 \cdot 10^{-3} \,  t_{pd} $, which roughly corresponds to $ 50 \, {\rm K} $ for $ t_{pd} = 1.3 \, {\rm eV} $. In order to determine the nature of this instability, we diagonalize the coupling functions $ \Phi_{\rm SC}$, $ \Phi_{\rm M}$ and $ \Phi_{\rm K} $
at the stopping scale. We attribute the largest of the eigenvalues of these three coupling functions to the leading instability, which is characterized by an optimized form factor given by the corresponding eigenvector.

For the parameter sets considered in the following, the most relevant eigenvalues in the pairing and the magnetic channels compete. Let us first describe our results for the generic parameter set of Tab.~\ref{tab:params}. 
The optimized form factors for this parameter set are depicted in Fig.~\ref{fig:genff}. 
In the pairing channel, contributions with total wavevector ${\bf l}=0$ dominate clearly. The optimized form factor corresponds to a $d_{x^2-y^2}$-wave, with peaks 
that are a little broader than for $ f_d = \cos (q_x) - \cos (q_y) $ (cf.\ Fig.~\ref{fig:genff}). We will comment on the admixture of higher harmonics to this optimized form factor further below. In the magnetic channel, the optimized form factor corresponds to a deformed $s$-wave with small admixtures of higher harmonics. For the generic parameter set we find slightly incommensurate magnetic ordering
vectors on the boundary of the Brillouin zone. Inequivalent IRs of $ C_{4v}$ may hence mix. In the case of our generic parameter set, we can indeed observe small admixtures of other IRs to the $s$-wave contributions of the optimized form factor. 
For example, a small $d_{x^2-y^2}$-wave admixture is clearly visible.

 Other types of instabilities such as
a Pomeranchuk instability\cite{halboth-metzner,pominst-wegner,neumayr-metzner,bulut-nematic,Husemann-pomer} and the formation of different types of loop currents\cite{varma,fischer,thomale} do not participate in the competition of the most relevant instabilities.
 For the former type of ordering, our results are in agreement with Ref.~\onlinecite{Husemann-pomer}, where the experimentally observed nematic tendency\cite{ando-nema-exp,daou-nema-exp,Hinkov-nema-exp2004,Hinkov-nema-exp2008} in
 cuprate materials appears to be of strong-coupling nature. Similarly, the absence of loop currents is not surprising in a weak-coupling scenario, since 
 the critical interaction strengths are found to be quite large in mean-field calculations.\cite{fischer}
 
At van-Hove filling, one can expect that it is possible to tune the Emery model to a Stoner-like ferromagnetic instability by raising the value of $ t_{pp} $. For the parameter sets considered in this work, however, we do not find ferromagnetism to prevail over other ordering tendencies. We hence conclude that a Stoner instability only occurs for oxygen-oxygen hopping parameters that are far away from the generic value of $t_{pp}$ in Tab.~\ref{tab:params}.

 At the stopping scale, we also consider a FFE of the full coupling function $V$ obtained within our new approach 
for comparison. We truncate this expansion behind the most slowly varying form factors. This is not to be confused with the results obtained from RG flow equations in an exchange parametrization, 
as we take mixing between different IRs of $ C_{4v} $ and of different harmonics within an IR into account in the integration of the flow equations. 
 The propagator $ D_d$ of $d_{x^2-y^2}$-wave Copper pairs rescaled by the square strength of the fermion-boson interaction then reads as 
\begin{align*}
 D_d (l) = \int d {\bf q} \, d {\bf q'} \, &f_d ( {l}/2 -{q} ) \,  f_d ( {l}/2 -{q'} ) \\ & \times \left.  V(q,l-q,q') \right|_{q_0=q_0'=l_0/2} \, , 
\end{align*}
where $ f_d (k) = \cos(k_x) - \cos (k_y) $. In the following, we will simply refer to $ D_d $ as a bosonic propagator despite the rescaling by an energy-squared factor. Likewise, the propagator of the magnetic $ s $-wave exchange-boson is obtained as
\begin{align*}
 M_s (l) = \int d {\bf q} \, d {\bf q'} \, &f_s ( {l}/2 +{q} ) \,  f_s (- {l}/2 +{q'} ) \\ & \times \left.  V(q,q',l+q) \right|_{q_0=-q_0'=-l_0/2} \, ,
\end{align*}
with $ f_s (k) = 1 $.
Note that this FFE is only viable if the basis of Bloch states is properly chosen such that $ V $ is invariant under all point-group transformations of the lattice.

In Fig.~\ref{fig:props-gen}, $D_d (l) $ and $M_s (l) $ are depicted for the generic parameter set in Tab.~\ref{tab:params} at the stopping scale.
 Both  $d_{x^2-y^2}$-wave Cooper-pair and magnetic $s$-wave
 propagators show peaks  with values close to the corresponding eigenvalues of the coupling functions of the respective channel. Since the $d$SC peak is quite sharp while the incommensurate peaks of the magnetic propagator have a broader width, $d_{x^2-y^2}$-wave superconductivity might prevail in a situation where the two most relevant instabilities are closely competing. 
 In any case, the system is in a regime of two competing, mutually reinforcing instabilities. 
 In this place, we would like to recapitulate that our stopping condition is quite weak. Therefore, if the magnetic and pairing channels are still of comparable strength at the stopping scale, these two channels are then indeed closely competing and this competition itself might have some physical content.

 In Ref.~\onlinecite{n-patch}, the parameter-space region of strong AFM-$d$SC competition was dubbed the \emph{saddle-point regime} and interpreted as an insulating spin-liquid phase.
It may also contain a region of homogeneous coexistence as described for the iron pnictides
 in Ref.~\onlinecite{schmalian}. It is likely that a large part of this regime has a non-vanishing superconducting gap. 
Unfortunately, order parameters are not directly accessible within the present instability analysis. In a recent fRG approach to the one-band Hubbard model using rebosonization techniques\cite{friederich} it has however been found that pairing is avoided inside the antiferromagnetic phase. This is not surprising, since at least parts of the Fermi surface are gapped away once spontaneous symmetry-breaking in one channel sets in, which hampers symmetry-breaking in the other channels.
A very recent purely fermionic study on the fRG flow of the repulsive single-band Hubbard model into the superconducting phase is in full agreement with this picture.\cite{eberlein-private} In that work, a non-vanishing pairing
gap is indeed found in a large subregion of the saddle-point regime.
 So a putative coexistence phase should be considerably smaller than the saddle-point regime.
 In principle, a two-order-parameter mean-field approach is viable below the stopping scale,\cite{rohe} but of course such a treatment is not free of bias.

As we shall find in the following, the character of the instability is quite robust against slight variations of the parameters. The stopping scale will turn out to be  sensitive to the chemical potential $ \mu $ and to the diagonal oxygen-oxygen hopping $ t_{pp} $ further below. This latter dependence is already rather smooth, if the oxygen-oxygen hopping
is not to far away from its generic value. So the behavior around the generic parameters suggests that the system is close to a first-order phase transition between AFM and $d$SC, as second-order transition would probably go along with a kink of the critical scale. 
 Since the RG stopping scale is an upper estimate for the critical scale, such a kink may however be hidden.
 For the one-band Hubbard model close to van-Hove filling, a similar behavior has been observed.\cite{n-patch,rohe,husemann_2009} 
\begin{figure*}
 \includegraphics[width=\linewidth]{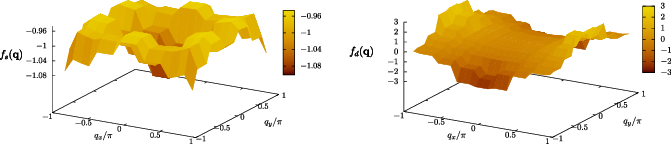}
 \caption{(Color online) Optimized from factors $ f_d$ and $ f_s $ as functions of the `fermionic' wavevector ${\bf q}$. These form factors are obtained as eigenvectors corresponding to the most relevant eigenvalues for the generic parameter set of Tab.~\ref{tab:params} at the stopping scale in the pairing and magnetic channels, respectively. Note that both optimized form factors are close to the most slowly varying basis functions of the respective irreducible representation of $ C_{4v} $. Due to the incommensurability of the ordering vector, $ f_s $ shows slight admixtures of other irreducible representations.}
 \label{fig:genff}
\end{figure*}
\begin{figure*}
 \includegraphics[width=\linewidth]{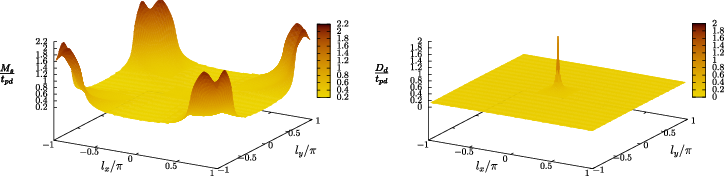}
	\caption{(Color online) Bosonic propagators at zero frequency corresponding to the most slowly varying basis functions for $s$-wave magnetism (left) and $d_{x^2-y^2} $-wave Cooper pairs (right) at the stopping scale for the parameters in Tab.~\ref{tab:params}, as functions of the bosonic wavevector ${\bf l}$.}
\label{fig:props-gen}
\end{figure*}
 \subsection{Doping dependence}
 The system also stays in the saddle-point regime when the doping level is slightly varied. In Fig.~\ref{fig:doping}, the stopping scale for the parameters in Tab.~\ref{tab:params} is plotted as a function of $ \mu $ ($+$-markers). In the following, we define the hole filling factor as
the number of holes per unit cell and spin orientation that have been doped into the originally half-filled conduction band 
\begin{equation} \label{eqn:def-nh}
 n_{\rm h} = \frac{1}{2} - \int \! d {\bf k} \, \Theta \left[ \mu - \epsilon ( {\bf k}) \right] \, ,
\end{equation}
where $ \epsilon ({\bf k}) $ and $ \Theta (x)$ denote the dispersion of the conduction band and the Heaviside 
step function, respectively.
 Note that $ n_{\rm h}$ may differ for the dispersions of the Emery model and the single-band models with effective parameters calculated according to 
Sec.~\ref{sec:eff-1bd}. In the following, all values of $n_{\rm h}$ will be for the Emery model.

The doping level varies between $ n_{\rm h} = 6.5 \cdot 10^{-2}$ and $ 14.3 \cdot 10^{-2} $ and at van-Hove filling the filling factor is $ 0.115 $. We observe that between half-filling ($n_{\rm h} = 0$) and van-Hove filling, the stopping scale only varies slightly. The AFM and $ d$SC instabilities are closely competing, except for $ n_{\rm h} < 0.09 $, where AFM clearly prevails. At hole doping beyond van-Hove filling, the stopping scale decreases rapidly and the tendency to $d_{x^2-y^2}$-wave pairing gets  a little stronger.
Qualitatively, this behavior is analogous to the hole-doped one-band Hubbard model.\cite{n-patch}
\begin{figure}
 \includegraphics[width=\linewidth]{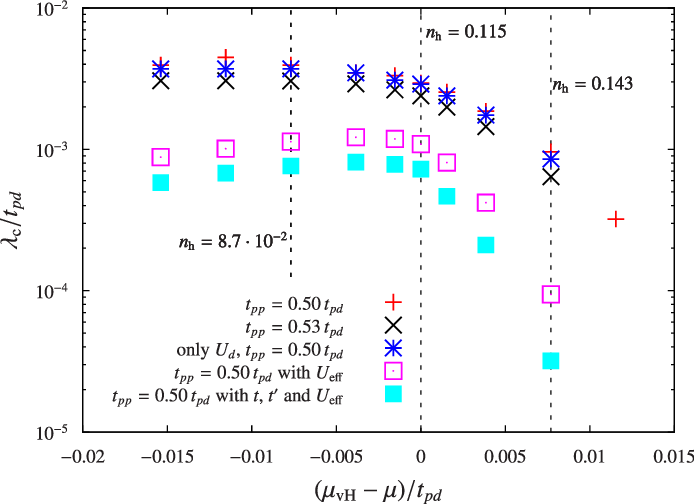}
	\caption{(Color online) Doping dependence of the stopping scale around van-Hove filling. The curve with +-markers is for the generic parameter set in Tab.~\ref{tab:params},
and the one with x-markers is for $ t_{pp} = 0.53 t_{pd} $ while all over parameters are left unchanged. The star-like markers depict the stopping scales for an interaction with a $U_d$-term only. The other two curves are for effective models: Empty squares represent data
for an effective  on-site Hubbard interaction with the full dispersion of the Emery model for the generic parameter set. Filled square markers are for the effective $t$-$t'$-$U_{\rm eff}$
Hubbard model.
 The corresponding hole-fillings $n_{\rm h}$ indicated by dotted vertical lines (see also Tab.~\ref{tab:fillings}) are for the Emery model (and not the single-band Hubbard model).}
\label{fig:doping}
\end{figure}
\begin{table}
 \begin{tabular}{l||cccccc}
 \hline
   $ \left(\mu_{\rm vH} - \mu \right) / t_{pd} \cdot 10^3 $ &  $ - 15.4 $ & $ -7.69 $ & $ -3.85 $ & $ 0 $ & $ 3.85 $ & $ 7.69 $ \\
\hline
   $  n_{\rm h} \cdot 10^2 $ & $ 6.5 $ & $ 8.7 $ & $ 9.9 $ & $ 11.5 $ & $ 13.1 $ & $ 14.3 $ \\
 \hline
 \end{tabular}
 \caption{Filling factors $ n_{\rm h} $ defined as in Eq.~(\ref{eqn:def-nh}) for the parameters set in Tab.~\ref{tab:params}.}
 \label{tab:fillings}
\end{table}
\begin{figure*}
 \includegraphics[width=\linewidth]{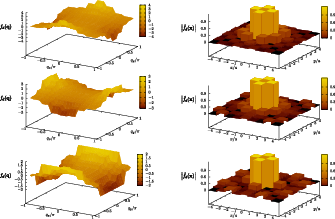}
	\caption{(Color online) Deformation of the optimized $d_{x^2-y^2}$-wave form factor $ f_d ({\bf q}) $ in the Cooper channel at the stopping scale (left column) and the absolute value of its Fourier components $ \hat{f}_d ({\bf x}) $ on the real lattice obtained from FFT (right column). The central row is for van-Hove filling and the upper and lower ones for $ n_{\rm h} = 8.7 \cdot 10^{-2} $ and for $ n_{\rm h} = 0.14 $, respectively.
All other parameters are chosen as in Tab.~\ref{tab:params}. The form factors have been normalized to $ \int \! d {\bf q} \, \left| f_d({\bf q}) \right|^2 =1 $.
  $f_d ({\bf q}) $ gets broadened at the saddle points with increasing hole doping.
  The (discrete) direct-space coordinate corresponds to the distance of two electrons forming a Cooper pair.
  For all three filling factors considered here, the main contribution to the pairing comes from electrons residing on neighboring sites, which corresponds to $ f_d = \cos (q_x)- \cos (q_y) $ in reciprocal space. Admixtures of higher harmonics are present in all three cases and get shifted away from the origin at hole doping beyond van-Hove filling.}
\label{fig:ffs}
\end{figure*}
\begin{figure}
 \includegraphics[width=\linewidth]{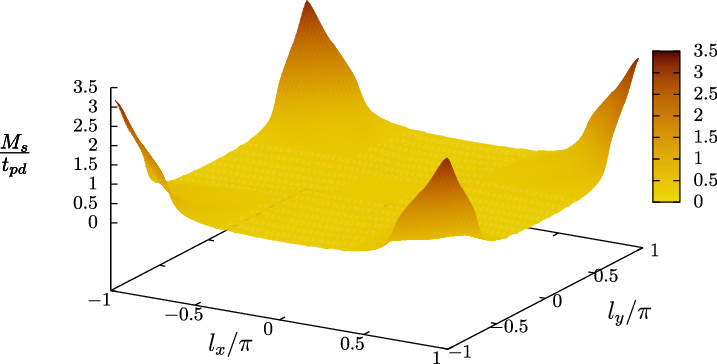}
	\caption{(Color online) Magnetic $s$-wave propagator as function of the bosonic wavevector ${\bf l}$ at the stopping scale for $ n_{\rm h} = 8.7 \cdot 10^{-2}$. All other parameters are chosen as in Tab.~\ref{tab:params}.}
\label{fig:incomm-0p005}
\end{figure}
\begin{figure}
 \includegraphics[width=\linewidth]{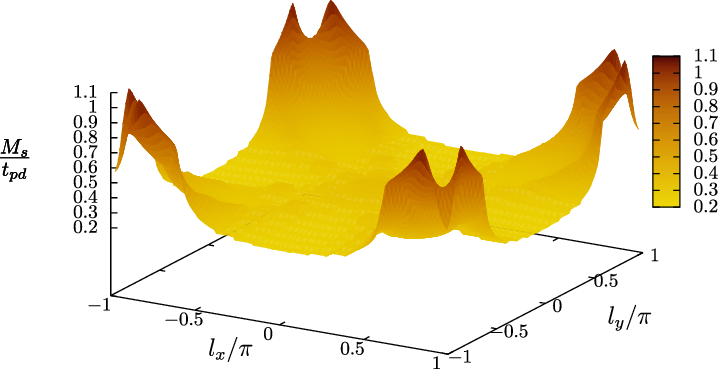}
	\caption{(Color online) Magnetic $s$-wave propagator as function of the bosonic wavevector ${\bf l}$ at the stopping scale for $ n_{\rm h} = 0.13$. All other parameters are chosen as in Tab.~\ref{tab:params}.}
\label{fig:incomm+0p002}
\end{figure}
\begin{figure}
 \includegraphics[width=\linewidth]{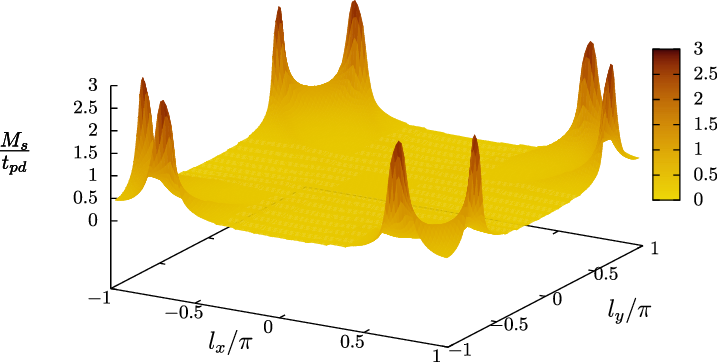}
	\caption{(Color online) Magnetic $s$-wave propagator as function of the bosonic wavevector ${\bf l}$ at the stopping scale for $ n_{\rm h} = 0.14$. All other parameters are chosen as in Tab.~\ref{tab:params}.}
\label{fig:incomm+0p010}
\end{figure}

The enhanced $d_{x^2-y^2}$-wave pairing at hole doping beyond van-Hove filling, however, comes with a broadening of the $d_{x^2-y^2}$-wave form factor at the van-Hove points as can be seen from Fig.~\ref{fig:ffs}.
Moreover, the magnetic $s$-wave propagator is of interest. In Figs.~\ref{fig:incomm-0p005}-\ref{fig:incomm+0p010}, we observe that for $ n_{\rm h} \leq 0.09$ the ordering vector $ {\bf l} = ( \pi,\pi) $ is commensurate. At about van-Hove filling, the commensurate peak of $ M_s $ is split into four peaks at incommensurate ordering
vectors. The deviation of the ordering vector from $ (\pi,\pi) $ then increases with hole doping, corresponding
to the shift of Fermi surface segments at high density of states. 

The highly incommensurate peaks at $ n_{\rm h} = 0.14 $ (i.e.\ at hole doping significantly beyond van-Hove filling) still allow the Kohn-Luttinger effect to generate an attractive $ d$-wave pairing component.
This however goes along with a deformation of the form factor $ f_d$ for $d$-wave pairing (see Fig.~\ref{fig:ffs}). Such a simultaneous occurrence of incommensurability in the magnetic channel and a deformation of the fermion-boson vertex in the Cooper channel has already been observed in the one-band $(t,t')$ Hubbard model\cite{jutta} and may be explained as follows.
Consider a singlet Copper pair with momenta $ ({\bf k},-{\bf k})$ scattered  to $ ({\bf k'},-{\bf k'})$ by the interaction in the Cooper channel which shall be mimicked by a one-loop particle-particle diagram with two spin-channel vertices. If these vertices have their peaks at transfer momentum $ {\bf l} = {\bf Q} = (\pi,\pi)$, the main contribution to the
Cooper channel comes from $ {\bf k'} = {\bf k} $. For incommensurate ordering vectors, the important contributions come from   $ {\bf k'} = {\bf k} $ as well as from $ {\bf k'} = {\bf k} + {\bf Q}_i + {\bf Q}_j $, where $ i$ and $j$ may correspond to all possible combinations of the ordering vectors. The dependence of fermion-boson vertex in the Cooper channel on the fermionic momenta is hence smeared out around $ (0,\pi ) $ and $ (\pi,0)$ resulting in a broadening of the form factor. A shoulder-like broadening of the peaks of $ M_s $ would give rise to the same effect in a 
similar way.

 While higher harmonics do not contribute at the so-called anti-nodal points $(0,\pi)$ and $(\pi,0)$, they may change the slope of the gap at the nodal points. 
 Such an occurrence of multiple energy scales for the gap has been observed in Raman spectroscopy\cite{leTacon-twogap} and angle-resolved photoemission spectroscopy\cite{ARPES-twogap}
 experiments. In contrast to our results for weak coupling, Ref.~\onlinecite{leTacon-twogap} suggests decreasing contributions of higher harmonics with hole doping.
 This may be due to the strong-coupling nature of real cuprate materials.
 
As the Fourier transform of the form factor in the pairing channel corresponds to the distribution of the distance between the electrons forming a Cooper pair, deviations from the $ \cos (q_x) - \cos (q_y) $-form may also be analyzed in real space (see lower row in Fig.~\ref{fig:ffs}). 
One should, however, be aware that an interpretation in real space requires some care, since
basis sets of Wannier functions may strongly differ in their localization properties (cf.\ Ref.~\onlinecite{ours-point-group}). 
The position argument in the real-space form factors then corresponds to the relative distance between the
 constituents of such a pair, i.e.\ the two electrons or holes involved.
 Note that this distance can be resolved up to only $n$ sites in all directions for $ 2n \times 2n $ fermionic patches.
 We have therefore studied the flow for some parameters with a resolution of $ 8 \times 8 $ fermionic and $ 24 \times 24 $ or $ 120 \times 120 $ bosonic patches away from and close to possible ordering vectors, respectively.
The results are displayed in Fig.~\ref{fig:ffs}. We observe that the most important contribution corresponds to a $ \cos (q_x) - \cos (q_y) $-form. But already at van-Hove filling, an admixture of higher harmonics is visible, which partly get shifted further away from the origin at hole doping beyond van-Hove filling. A thorough discussion of the minor contributions corresponding to Cooper pairing beyond nearest neighbors may require a resolution higher than $ 8 \times 8 $ fermionic patches.  

Before we analyze the impact of the $t_{pp}$-hopping parameter, a remark on the effects of the coupling between the different channels seems to be in order. These effects go far 
beyond the spin-fluctuation induced generation of an attractive pairing interaction. In particular, if the magnetic propagator was calculated within RPA, i.e.\ if the Cooper and forward scattering channel were neglected in the flow, the
stopping scale would be about one decade higher. Moreover, the magnetic propagator would be less sharply peaked. This behavior can be attributed to the feedback of the Cooper on the magnetic channel, which hampers antiferromagnetism before the $d_{x^2-y^2}$-wave pairing interaction gets attractive. On a qualitative level, this effect is already captured in the two-patch approximation.\cite{furukawa}
\subsection{Dependence on hopping-between the $p$-orbitals} \label{sec:tpp}
So far, we have only investigated the impact of doping away from van-Hove filling, but not the interplay between
$ \mu $ and $ t_{pp} $.
  If, in this spirit, the hopping between the $ p$-orbitals is now changed to $ t_{pp} = 0.53 \, t_{pd} $, the effective second-neighbor hopping in the conduction band gets stronger and the tendency to $d_{x^2-y^2}$-wave pairing should be 
enhanced. Indeed, the flow can now be attributed to the saddle point regime for all filling factors considered. The corresponding curve in Fig.~\ref{fig:doping} (x-markers), however, looks similar to the one for $ t_{pp} =  0.50 \, t_{pd} $ except for the insignificantly lower stopping scale. 
Moreover, at hole doping beyond van-Hove filling and  at $ t_{pp} =  0.53 \,  t_{pd} $, the tendency to $d$SC is only slightly enhanced compared to $ t_{pp} = 0.5 \, t_{pd} $.
It therefore seems that a considerable region of the parameter space has to be attributed to the saddle-point regime as for the one-band Hubbard model in Ref.~\onlinecite{n-patch}.
 \begin{figure}
 \includegraphics[width=\linewidth]{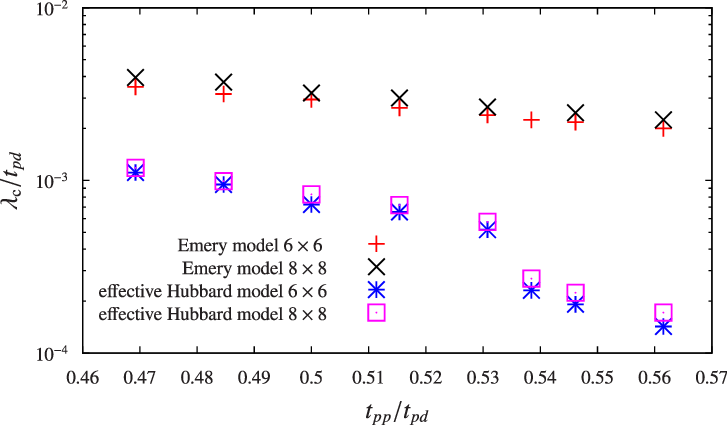}
\caption{(Color online) Variation of the stopping scale with $ t_{pp} $ at van-Hove filling. All other parameters are chosen according to Tab.~\ref{tab:params}. The $d_{x^2-y^2}$-wave pairing tendency increases with $  t_{pp} / t_{pd} $.} 
\label{fig:tpp}
\end{figure}
 
In Fig.~\ref{fig:tpp}, the dependence of the stopping scale on $ t_{pp} $ is depicted for van-Hove filling both for $ 6 \times 6 $ and for $ 8 \times 8 $ fermionic patches. The curves for the two resolutions
almost coincide, indicating that $ 6 \times 6 $ fermionic patches are sufficient. The stopping scale behaves as follows:

For the Emery model, we find a decrease of $ \lambda_{\rm c} $ with increasing  $ t_{pp} $ of less than one decade. Such a behavior is quite generic as, in the absence of orbital makeup, a more rounded Fermi surface depresses the stopping scale in other models.\cite{uebelacker} Once  $ t_{pp} $ exceeds $ 0.54 \, t_{pd} $, the flow to strong coupling bears rather the signatures of pure $d_{x^2-y^2}$-wave pairing than of the saddle-point regime. For the one-band Hubbard model with effective parameters, however, the situation is different: 
 The stopping scales are much lower and drop abruptly as soon as the system enters the pure $d$SC regime at about
$ t_{pp} = 0.54 \, t_{pd} $. This is apparently caused by an abrupt growth of the hybridization at the saddle points which reduces the effective interaction strength $ U_{\rm eff} $. 
Comparing the Emery model to its effective one-band Hubbard counterpart, we find that orbital-makeup effects in the Emery model partly counteract the decrease of $ \lambda_{\rm c} $ with a more rounded Fermi surface and that they prevent the stopping scale from dropping abruptly.

 In Ref.~\onlinecite{uebelacker}, a similar behavior as been found for multiband models only involving orbitals on the Copper atoms.
 Also the DCQMC results of Ref.~\onlinecite{Kent-emery-model} support this conjecture. Namely, the critical temperature is found to increase with the value of $t_{pp}$ in that work. At first glance, this seems to contradict our fRG results. However, the calculations in Ref.~\onlinecite{Kent-emery-model} have been performed
 at strong coupling, where the orbital makeup seems to overcompensate the effect of the more rounded Fermi surface observed at weak coupling.
 Of course, this argument is not fully stringent, since, at strong coupling, mechanisms may be at work that do not occur in a perturbative picture.
 \begin{figure}
 \includegraphics[width=\linewidth]{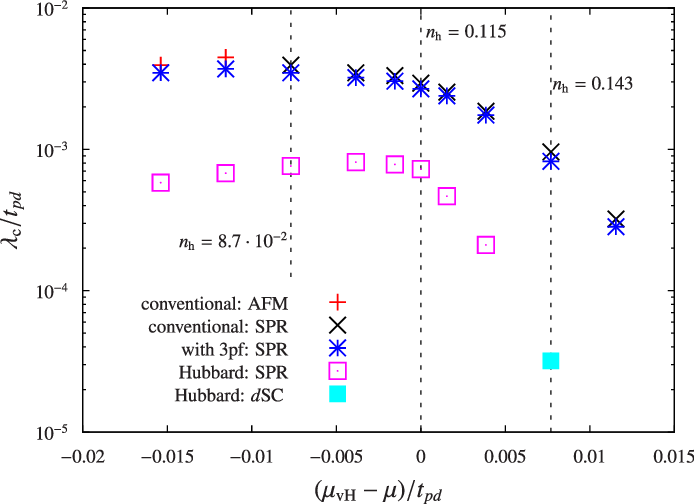}
\caption{(Color online) Phase diagrams for the generic parameter set of Tab.~\ref{tab:params} with and without three-particle feedback (3pf) and for the corresponding one-band Hubbard model. Large parts of these phase diagrams can be attributed to the saddle-point regime (SPR).
 As in Fig.~\ref{fig:doping}, the corresponding hole-fillings $n_{\rm h}$ indicated by dotted vertical lines (see also Tab.~\ref{tab:fillings}) are for the Emery model (and not the single-band Hubbard model).}
\label{fig:phase-diag}
\end{figure}

 \subsection{Three-band vs.\ single-band model} \label{sec:eff-short-range}
Let us now look at multiband effects more systematically.
So far, we have discussed results for the full Emery model as given by the parameters in Tab.~\ref{tab:params} and variations of $ t_{pp} $ and $ \mu $ in a conventional truncation of the flow equations, i.e.\ without a three-particle term. If the feedback of this term is taken into account within the
truncation proposed in Ref.~\onlinecite{3particle}, we do only find insignificant changes of the stopping scale.
Moreover, the tendency to $d$SC is slightly enhanced between half-filling and van-Hove filling, so that the system stays in the 
saddle point regime down to $ n_{\rm h} = 6.5 \cdot 10^{-2} $, as can be seen from the phase diagrams depicted in Fig.~\ref{fig:phase-diag}.
 Such minor modifications of the phase diagram seem quite surprising as those feedback terms had a great impact
on the phase diagram of a two-band model in a two-patch approximation. Such a behavior of the three-particle feedback 
may, however, may be an artifact of the two-patch approach pursued in that older study.
Namely, the two-patch approximation only allows for a
small number of strong-coupling fixed points which results in a mutual exclusion of the Cooper and Stoner instabilities. Moreover, in the Emery model, the large gap between the conduction band and the valence bands results in a flat momentum structure of the diagrams corresponding to three-particle feedback terms. 
 In a frequency-resolved study, the three-particle feedback may however play a more important role.

 Before we elaborate on other multiband effects, the role of the oxygen $p$-orbitals in the two-particle term shall be discussed. 
 First, we turn our attention to the impact of a $ U_{pp} $-term, which should be absent according to Ref.~\onlinecite{hybertsen}. We have varied $ U_{pp} $ from zero to $ 0.1 \, U_d $. The stopping scale then only changes insignificantly and the systems stays in the saddle-point regime.

In a second step, the role of the interaction terms involving legs on the $ p$-orbitals is now discussed, i.e.\ we compare the result for the full Emery model to what is obtained
if all interaction terms except $ U_d $ are ignored. Note that this interaction is still non-local in the band language, and that it is therefore different from an on-site Hubbard term. 
Data points for this level of approximation with dispersion parameters according to Tab.~\ref{tab:params} are represented by stars in Fig.~\ref{fig:doping}. These points almost coincide
with the ones for the full interaction ($+$-symbols). Therefore interaction terms involving the $p$-orbitals only slightly affect the stopping scale.

However, this does not mean that these terms do not have any impact at all. Actually, the tendency to AFM is enhanced if the $ p$-orbital interaction terms are neglected. In particular, between half filling and van-Hove filling,
the flow to strong coupling shows features of an AFM-instability rather than of the saddle-point regime. A form factor deformation above
van-Hove filling still occurs, but this effect is a little weaker without $ U_{p} $ and $ U_{pd} $.
One may now try to understand the enhancement of the AFM tendency by considering the two-patch couplings $ g_i $ in the ultraviolet. For the full interaction $ g_1 $, $ g_2 $ and $g_4 $ have the same value of about $ 0.269 \, t_{pd} $ while $ g_3 $ is lowered by a few percent to $ 0.264 \, t_{pd} $. For an interaction involving the $ d$-orbital only, in contrast, $ g_1 = g_4 = g_2 = 0.255 \, t_{pd} $  and $ g_3 $  is now enhanced to $ 0.257 \, t_{pd} $. So we have an overall decrease of the two-patch couplings and their detuning differs from the case of full interaction,
resulting in a relative increase of the $ d$-wave pairing and the AFM components of the interaction in the two-patch approximation, $ g_3 - g_4 $ and $ g_1 + g_3 $, respectively. Apparently, a lowering of the stopping scale induced by the former is compensated by the latter and, as a net effect, the AFM tendency gets stronger.

We now continue with the discussion of the other multiband effects that have been listed in Sec.~\ref{sec:eff-1bd}. 
First, we replace the two-particle interaction for the conduction band  by an on-site Hubbard interaction while we leave the dispersion unchanged. Its strength $ U_{\rm eff} $ is chosen to be the average of the two-patch couplings that correspond to the full interaction.
From the open squares in Fig.~\ref{fig:doping}, we infer that the stopping scale is significantly lowered in this approximation.
As the detuning of the two-patch couplings is rather small on this level of approximation, we conclude that phase-space regions away from the saddle points
play some (minor) role even at van-Hove filling. Since the hybridization of $p$- and $d$-orbitals is strongest at the saddle points, the contributions of those regions to the diagrams on the right-hand side of the flow equations are underestimated by an on-site interaction with strength $ U_{\rm eff} $. Therefore a flow to strong coupling occurs at lower scales. 

One may therefore wonder, whether another prescription for choosing $ U_{\rm eff} $ might give results that are more close to those for the three-band model.
Since the hybridization between $d$ and $p$-orbitals is strongest at the saddle points and since the $U_p$ and $U_{pd}$ interactions are of minor importance, the choice $ U_{\rm eff} = U_d $ seems appealing as well.
If one were to follow this alternative prescription,
the stopping scale of the effective model would overshoot the value for the original model by roughly a factor of two.
 This is not surprising, since the averaged interaction strength of the two models then already differs at the saddle points. Moreover, the prescription $ U_{\rm eff} = U_d $ can be regarded as the leading-order result of a gradient expansion around $ {\bf k} = (0,0) $. However, this is inconsistent with the expansion around the saddle points of the conduction band underlying the calculation of the effective hopping parameters $ t$ and $t'$. We therefore continue to use the prescription given in Sec.~\ref{sec:eff-1bd}, since the approach seems to be the most systematic one.

 The suppression of orbital makeup then generically lowers the stopping scale. Away from van-Hove filling (in particular at hole doping beyond), this lowering is more pronounced as the gradient expansion gets worse. 
Moreover, the system stays in the saddle-point regime for all filling factors considered. There are now several possible mechanisms giving rise to the enhanced tendency to $d$SC between half filling and van-Hove filling.
As the attractive $d_{x^2-y^2}$-wave pairing component is generated by fluctuations in the magnetic channel, this enhancement may simply be caused by the lowering of the stopping scale.
Furthermore, a detuning of the two-patch couplings that hampers the $d$SC instability is now absent. 
Let us note in passing that the broadening of the $d_{x^2-y^2}$-wave form factors around the saddle points is still restricted to hole doping beyond van-Hove filling.

Finally, we consider the $ t,t',U_{\rm eff} $ one-band Hubbard model (filled squares in Fig.~\ref{fig:doping}), i.e.\ we now approximate also the dispersion by the leading-order result
of a gradient expansion around the saddle points. Compared to the previous data set, the stopping scale is again lowered by almost a factor of two at van-Hove filling. Away from van-Hove filling, this depression of $ \lambda_{\rm c} $ again grows. At hole doping beyond van-Hove filling, this effect is more pronounced between half filling and van-Hove filling and at about
$ n_{\rm h} = 0.14 $ the system enters the pure $d$SC regime. 
Again, a form factor deformation occurs at van-Hove filling and larger hole doping. Compared to the original model, the stopping scale is a factor between five and ten too low. We therefore conclude that in the effective action for the conduction band long-range hopping terms play a significant role, since they enhance the stopping scale.

\section{Discussion and Outlook} \label{sec:conclusion}
 In this article, we have studied the hole-doped three-band Emery model at weak coupling by considering the RG flow of an effective interaction for the conduction band. This has been done within an improved channel-decomposed approach, which considerably improves the 
 momentum dependence of fermion-boson vertices of previous treatments of the Hubbard model within an exchange parametrization.
 In Sec.~\ref{sec:FFE}, we have classified the contributions that are present in our approach, but potentially neglected in Refs.~\onlinecite{husemann_2009,HGS-freq-dep,giering-new}, albeit without qualitative changes of the leading instabilities. As new results we present tentative phase diagrams for the Emery model, the wavevector structure of the effective interactions, in particular in the spin channel,  and the structure of the (deformed) $d$-wave pairing gap. We also provide a detailed comparison between the RG flow in the three-band case in one-band models with comparable Fermi surfaces.
 
In the Emery model for the parameters considered here, the leading instabilities are $d_{x^2-y^2}$-wave pairing and antiferromagnetism. The leading instability often can hardly be distinguished from the subleading one, as also found in the one-band case\cite{n-patch}. 
For most of the parameter sets considered,
the system appears to be in some intermediate region between the regimes with clear AFM and $d$SC instabilities, where both tendencies are comparably strong even very close to the critical scale. This so-called saddle point regime might contain a coexistence phase and should be superconducting to a large part. Signatures of this saddle-point behavior survive at doping levels a few percent away from van-Hove filling, again similarly to Ref.~\onlinecite{n-patch}. 
Between half filling and van-Hove filling, the stopping scale shows a plateau. 
Once the system is hole-doped beyond van-Hove filling, the stopping scale decreases rapidly. This goes along with an increasing incommensurability of the AFM ordering vector giving rise to a broadening of the superconducting gap at the saddle-points of the dispersion. In real space, this corresponds to Cooper pairing between electrons in
Wannier states that are centered around points in non-neighboring unit cells.
 Such effects could not be explored with previous $N$-patch fRG schemes, but are now available by virtue of the refined wavevector resolution of the channel-decomposed vertices. They are not intrinsically of
multiband nature, and can also be observed in the single-band model.
All these properties are quite robust against the omission of interaction terms involving the oxygen $p$-orbitals.

The effective second-nearest-neighbor hopping $t'$ of the conduction band can be tuned by changing the oxygen-oxygen hopping $ t_{pp} $ of the Emery model. An increased value of this parameter therefore leads to a more rounded Fermi surface at and close to van-Hove filling.
 We observe that the stopping scale then decreases only slowly with increasing $ t_{pp} $. 
For the corresponding single-band Hubbard model, we find a decrease of the stopping scale with $ t_{pp} $ as well, but now with a peculiar feature: Between $ t_{pp} = 0.53 \, t_{pd} $ and $  0.54 \, t_{pd} $, where the
copper-oxygen hybridization in the full Emery model rises abruptly at the van-Hove points, the stopping scale suddenly drops by a factor of two.
Therefore, orbital-makeup effects included in the three-band model apparently counteract and almost compensate the effect of a more rounded Fermi surface, similarly to what was found in Ref.~\onlinecite{uebelacker}. In other words, not only the Fermi surface shape and the density of states matters for the energy scale of those weak-coupling instabilities.

Close to van-Hove filling, it seems reasonable to determine the parameters of the effective one-band Hubbard Hamiltonian in Eq.~(\ref{eqn:hubbard}) from a gradient expansion around the saddle points of the dispersion, as the $ C_{4v} $ symmetry shared by both models in a suitable Bloch basis allows for such a procedure. In order to obtain a single parameter for the interaction, the 
effective on-site interaction is chosen as the average of the four two-patch couplings, which appear in leading order
 in a gradient expansion and would constitute the running couplings in the two-patch approximation. Except for a detuning of the two-patch couplings, the single-band and the three-band model coincide at the saddle points of the dispersion.

We find that the differences away from the saddle points are crucial in the sense that the Emery model formulated in the band language contains hopping and interaction terms connecting non-neighboring unit cells, which turn out to enhance the stopping scale.
Yet on a qualitative level, the one-band Hubbard model still has a similar phase diagram.
 On a more quantitative level, these longer-ranged terms play a role both in the dispersion, where 
hopping between non-neighboring unit cells has some impact, and in the interaction, which is decorated by orbital makeup. For the effective Hubbard model, the lowering of the stopping scale enhances the tendency to $d_{x^2-y^2}$-wave pairing, whereas the detuning of the two-patch couplings in the original model is of minor importance.
This is in analogy to recent VCA results\cite{Hanke-Kiesel-VCA} for the strong coupling case, where the effective one-band Hubbard model also accounts for the more universal features of the phase diagram.

Of course, a different prescription for calculating the parameters of the effective one-band Hubbard model could have been chosen in the present work. The one presented, however, appears to be rather systematic. It allows one to determine the origin of the deviations from the original model which should be related to ${\bf k}$-space regions away from the saddle points. More precisely, by using the interaction values a the saddle points for the effective one-band model, the interaction away from the Fermi surface and away from the saddle points is underestimated. This happens because the hybridization of  the $d_{x^2-y^2}$-orbital with the $p$-orbitals is largest at the saddle points. The relatively strong impact of the other regions on our results therefore 
suggests that those degrees of freedom are important in the sense that their presence does not only allow for a rich fixed point structure of the RG flow, but that they also directly influence the phase diagram.


A significantly large region of the parameter space constitutes the saddle point regime, where the $d$SC and AFM instabilities are closely competing and mutually reinforcing. The present instability analysis lacks directly accessible measurable quantities in potentially symmetry-broken phases which would facilitate the interpretation of such a behavior. In order to find a sharp phase boundary or a coexistence phase between AFM and $d$SC, it would be advantageous to enter the symmetry-broken phases within a purely fermionic approach. Furthermore, other, more exotic possibilities such as a Fermi surface truncation without long-range order\cite{furukawa,n-patch} are also inaccessible without proper flows for the self-energy.
Including these effects together is a formidable task beyond the current frontier in RG methods. We therefore have to refrain from explicitly breaking the U(1) and/or SU(2) symmetries, mentioning that fermionic flows have been continued into the  superfluid phase\cite{brokensy_SHML,gersch_superfluid,eberlein_param,eberlein-sssu,eberlein-private} and that
we have addressed parametrization questions and mean-field models for the AFM phase in a recent publication,\cite{ours-recent} while we are currently studying the AFM phase beyond mean-field models within a channel-decomposed approach. Moreover, as recently suggested by Giering and Salmhofer,\cite{giering-new} the parameters of an effective partially bosonized theory may be derived within a purely fermionic RG flow, whereas symmetry-broken phases are entered within a mixed flow. 

Such issues left aside, the strong-coupling nature of ab-initio parameter sets still prevents us from thoroughly discussing the applicability of our weak coupling approach (with interaction parameters one decade smaller than typical literature values) to real cuprate materials.
 We therefore refrain a discussion of some parameter trends observed in the strong-coupling literature (see, for examples, Refs.~\onlinecite{weber-charge-trans,Kent-emery-model}).
 However, the Emery model at weak coupling is shown to have the same leading instabilities than the one-band Hubbard model, complementing a strong-coupling VCA study by Kiesel \emph{et al.}\cite{Hanke-Kiesel-VCA} and a DMFT study by de' Medici \emph{et al.}\cite{Medici-emery-model}
Moreover, it has turned out to form a good, rather simple testbed for the new discretization scheme presented here. This scheme  can of course be carried over to other, more complicated multi-band systems where a weak-coupling approach is indeed realistic. Candidates are the iron superconductors or strontium ruthenates (for a recent SMFRG study, see Ref.~\onlinecite{Sruo}), both with three and more Fermi surfaces.


\section*{Acknowledgments}
 The authors thank A.~Eberlein, K.-U.~Giering, C.~Husemann, A.~A.~Katanin, T.~C.~Lang,  C.~Platt, M.~Salmhofer, M.~Scheb, R.~Thomale, S.~Uebelacker and Q.-H. Wang for fruitful discussions. We are also glad to acknowledge the friendly support provided by the
 high-performance computing team at RWTH. This work was supported by the DFG priority program SPP1458 on iron pnictide superconductors and by the DFG research unit FOR 723 on functional renormalization group methods. 

\begin{appendix}
\section{Point-group symmetries of the effective one-band action}
\label{sec:eff-pointgroup}
Although this might be hard to see at first glance, the Emery model possesses a hidden $ C_{4v} $-symmetry. In Ref.~\onlinecite{ours-point-group}, 
  two of us have discussed how point-group symmetries manifest themselves in a large class of multiband models.
 Also the Emery model belongs to that class. Namely,  there exists a three-dimensional representation
$ \{ M_{\hat{O}} \}_{\hat{O} \in C_{4v}} $ such that the Emery Hamiltonian (including the interaction terms) is invariant under the substitution 
\begin{equation*}
 \Psi_\sigma (\mathbf{k}) \overset{\hat{O}}{\longrightarrow} M_{\hat{O}} \, \Psi_\sigma (R_{\hat{O}} \mathbf{k}) \quad \forall \, \hat{O} \in C_{4v} \, .
\end{equation*}
 Here, the $ R_{\hat{O}} $ denote the two-dimensional rotation matrices, which form a faithful representation of $ C_{4v} $. 
 For a coordinate exchange $ \hat{I}' $ and a reflection $ \hat{I} $ with respect to the $ x $-axis,
 the respective three-dimensional representation matrices read as
\begin{equation*}
  M_{\hat{I}} ({\bf k}) = \left( \begin{array}{rrr}  1 & 0 & 0\\
                        0 & -1 & 0 \\
			0 & 0 & 1 \end{array} \right) \, ,  \quad
  M_{\hat{I}'} ({\bf k}) = \left( \begin{array}{rrr}  1 & 0 & 0\\
                        0 & 0 & 1 \\
			0 & 1 & 0 \end{array} \right) \, .
\end{equation*}
Since all point-group operations $ \hat{O} \in C_{4v} $ can be written as products of the identity operation, $ \hat{I} $ and $\hat{I}'$, the other $ M_{\hat{O}} $ follow from the group law $ M_{\hat{A}} = M_{\hat{C}} \,
M_{\hat{B}} $ for $ \hat{A} = \hat{B} \hat{C} $. 
Note that the representation matrices $M_{\hat{O}} $ decay into irreducible blocks --- a one-dimensional $A_1$ block for the Copper $d$-orbitals and a two-dimensional $E$ block for the oxygen $p$-orbitals.
 
In our fRG approach, the model is expressed in a band language, i.e.\ in terms of new pseudo-spinors
\begin{equation*}
 \bm{\chi}_\sigma (\mathbf{k}) = u (\mathbf{k}) \, \Psi_\sigma (\mathbf{k})
\end{equation*}
with a unitary, wavevector-dependent $ u (\mathbf{k}) $, where $ u(\mathbf{k}) \, \mathcal{H}_0 (\mathbf{k}) \,
 u^\dagger (\mathbf{k}) $ is diagonal. Of course, we would like to exploit the hidden $ C_{4v} $-symmetry of the model in the numerical integration of the RG flow equations.
 As shown in Ref.~\onlinecite{ours-point-group}, this can be accomplished as follows.
 There is some freedom in the transformation from orbitals to bands, corresponding to different choices of the phase of the eigenvectors in the
 lines of $ u (\mathbf{k}) $. These phases can be chosen individually for each wavevector $ \mathbf{k} $ and for each band. Since in the second-quantized language $\mathbf{k}$-dependent field operators create Bloch states from the vacuum state, this phase plays the role of a global phase of these Bloch states.
 In a so-called \emph{natural Bloch basis,} one has
 \begin{equation*}
  u ( R_{\hat{O}} \mathbf{k}) = u (\mathbf{k}) \, M_{\hat{O}}  \quad \forall \, \mathbf{k} \neq R_{\hat{O}} \mathbf{k} \, .
 \end{equation*}
Consequently, expressed in such a basis, the Hamiltonian is invariant under
\begin{equation*}
 \bm{\chi}_\sigma (\mathbf{k}) \to \bm{\chi}_\sigma (R_{\hat{O}} \mathbf{k}) \quad \forall \, \hat{O} \in C_{4v} 
\end{equation*}
and hence  the coupling function  $ U (k_1,k_2,k_3) $ in Eq.~(\ref{eqn:parameter}) shows a trivial point-group behavior, i.e.
\begin{equation*}
 U (R_{\hat{O}} \mathbf{k}_1,R_{\hat{O}} \mathbf{k}_2,R_{\hat{O}} \mathbf{k}_3) =
 U(\mathbf{k}_1,\mathbf{k}_2,\mathbf{k}_3)  \, .
\end{equation*}
This property is as well inherited by the coupling functions $ V(k_1,k_2,k_3) $  of the \emph{renormalized} interaction in Eq.~(\ref{eqn:cond-par}) and $ U_\beta (k_1,k_2,k_3) $ which enters the three-particle feedback in Eq.~(\ref{sec:3p-feedback}).

In this paper, we therefore work in a natural Bloch basis. According to Ref.~\onlinecite{ours-point-group},
discontinuities in the interaction then only occur in the terms involving the $p$-orbitals, which will turn out to be of minor importance.

\section{Flow equations}
In this Appendix, the flow equations for the two-particle interaction are given in a parametrization that exploits
U(1), SU(2) and translational symmetries. The quartic part $ \Gamma^{(4)} $ of the (scale-dependent) one-particle irreducible functional in the conduction band is parametrized in terms of a coupling function
$ V(k_1,k_2,k_3) $ according to
\begin{align} \notag
 \Gamma^{(4)}  [\bar{\chi}_-,\chi_-] &= - \frac{1}{4} \int \prod_i d (\sigma_i,k_i) \,  \bar{\chi}_{-,\sigma_1} (k_1) \, \bar{\chi}_{-,\sigma_2} (k_2) \\ \notag
 & \quad \times  \chi_{-,\sigma_3}  (k_3) \, \chi_{-,\sigma_4} (k_4) \,\delta (k_1 +k_2 -k_3 -k_4)\\ \notag
  & \quad \times  \left[ V (k_1,k_2,k_3) \, \delta_{\sigma_1, \sigma_4}
 \delta_{\sigma_2, \sigma_3} \right. \\ \label{eqn:cond-par}
 & \quad  \left. - V (k_2,k_1,k_3) \, \delta_{\sigma_1, \sigma_3} \delta_{\sigma_2, \sigma_4} \right] \, .
\end{align}

\subsection{Conventional truncation}
\label{sec:flow-eq}
In the conventional truncation, three-particle and higher vertices are neglected in the flow. (For a derivation see Refs.~\onlinecite{salm_hon_2001,Metzner-Rmp}.) The scale derivative of the coupling function $ V $ consists of three parts
\begin{align*}
 \partial_\lambda  V (k_1,k_2,k_3) & = {\cal T}_{\rm pp} (k_1,k_2,k_3)
 +  {\cal T}^{\rm cr}_{\rm ph}  (k_1,k_2,k_3) \\ & \quad + {\cal T}^{\rm d}_{\rm ph} (k_1,k_2,k_3) \, .
\end{align*}
The particle-particle contribution 
\begin{align*}
 {\cal T}_{\rm pp} = - & \int \! d p \, \left[ \partial_\lambda G(p) \, G(k_1+k_2-p) \right] \\
\times & V (k_1,k_2,p) \, V(k_1+k_2-p,p,k_3)
\end{align*}
and the crossed particle-particle part 
\begin{align*}
 {\cal T}^{\rm cr}_{\rm ph} = - & \int \! d p \, \left[ \partial_\lambda G(p) \, G(p+k_3-k_1) \right] \\
\times & V (k_1,p+k_3-k_1,k_3) \, V(p,k_2,p+k_3-k_1)
\end{align*}
can each be represented by one diagram with one loop containing (scale-dependent) propagators $ G (k) $ on the conduction band.
Since self-energy effects are neglected, we have
\begin{equation*}
 G(p) = R_\lambda (p_0) \, \left[i p_0 - \epsilon_1 ({\bf p}) \right]^{-1} \, ,
\end{equation*}
where $ \epsilon_1 $ denotes the energy of the conduction band and where the $ \Omega$-scheme regulator $ R_\lambda $ is given in Eq.~(\ref{eqn:Omega-reg}).
  Note that an RPA resummation in the Cooper or particle-hole channel is 
equivalent to an RG flow in which all terms except $ {\cal T}_{\rm pp} $ or $ {\cal T}^{\rm cr}_{\rm ph} $, respectively, are neglected. (The corresponding Bethe-Salpether equation is then equivalent to the flow equation for the two-particle vertex.\cite{husemann_2009}) Vertex corrections and screening, however, are accounted for by the direct particle-hole diagrams
\begin{align*}
 {\cal T}^{\rm d}_{\rm ph} =  & \int \! d p \, \left[ \partial_\lambda G(p) \, G(p+k_3-k_1) \right] \\
\times & \left[ 2 V (k_1,p+k_2-k_3,p) \, V(p,k_2,k_3) \right. \\
& -  V (k_1,p+k_2-k_3,k_1+k_2-k_3) \, V(p,k_2,k_3)  \\
& \left. -  V (k_1,p+k_2-k_3,p) \, V(p,k_2,p+k_2-k_3) \right] \, .
\end{align*}
If all external frequencies are projected to zero, the Matsubara sums in the loops can straightforwardly be evaluated analytically within the $ \Omega$-scheme using contour techniques.

The initial condition $ V=U $ is given by the quartic part 
\begin{align*}
 S^{(4)} [\bar{\psi},\psi] = & \, - \frac{1}{4} \sum_{\{X_i\}} F(X_1,X_2,X_3,X_4) \\
& \times \bar{\psi}(X_1) \, \bar{\psi}(X_2) \,  \psi(X_3) \, \psi (X_4)
\end{align*}
of the bare action for all three orbitals, where $ X_i = (k_i,\sigma_i,\alpha_i) $ is a short-hand notation for the 
quantum numbers of the fields and where the $ \alpha_i $ denote orbital indices. The coupling function $ F $ may be parametrized as in Eq.~(\ref{eqn:cond-par})
\begin{align*}
  F(X_1,X_2,X_3,X_4) & = 
  \left[ W^{\boldsymbol \alpha} (k_1,k_2,k_3) \, \delta_{\sigma_1, \sigma_4} \delta_{\sigma_2, \sigma_3}  \right. \\
& \quad - \left. W^{\tilde{\boldsymbol \alpha}} (k_2,k_1,k_3) \, \delta_{\sigma_1, \sigma_3} \delta_{\sigma_2, \sigma_4} \right] \, .
\end{align*}
with $ {\boldsymbol \alpha} = (\alpha_1,\alpha_2,\alpha_3,\alpha_4) $ and
$ \tilde{\boldsymbol \alpha} = (\alpha_2,\alpha_1,\alpha_3,\alpha_4) $.
We therefore have 
\begin{align*}
 U(k_1,k_2,k_3)  = \sum_{\boldsymbol \alpha} & W^{\boldsymbol \alpha} (k_1,k_2,k_3) \,\, u^\ast_{1,\alpha_1} (k_1) \,u^\ast_{1,\alpha_2} (k_2) \\
 & \times u_{\alpha_3,1} (k_3) \,u_{\alpha_4,1} (k_1+k_2-k_3) \, ,
\end{align*}
where $ u_{\alpha,\beta} (k) $ denote the matrix elements of the orbital-to-band transformation.
\subsection{Three-particle feedback}
\label{sec:3p-feedback}
 In the band picture, the two-particle term of the bare action also contains terms with three legs on the conduction band and one 
 on the valence band with band index $\beta$. In that case, the corresponding coupling function reads as
\begin{align*}
 U_{3,\beta}(k_1,k_2,k_3)  = \sum_{\boldsymbol \alpha} & W^{\boldsymbol \alpha} (k_1,k_2,k_3) \,\, u^\ast_{\beta,\alpha_1} (k_1) \,u^\ast_{1,\alpha_2} (k_2) \\
 & \times u_{\alpha_3,1} (k_3) \,u_{\alpha_4,1} (k_1+k_2-k_3) \, .
\end{align*}
Terms that have the valence-band index on another leg than the first can be reconstructed from this coupling function
by exploiting the Pauli principle and the particle-hole symmetry of the interaction.
When the valence bands are integrated out, a three-particle term is generated. In this Appendix, we will only give 
the resulting flow equation for the three-particle feedback and omit intermediate steps. (For a more detailed discussion of the case of one valence band, we refer to Ref.~\onlinecite{3particle}, as the generalization to a larger number of valence bands is straightforward.)

These three-particle feedback terms are now given separately as corrections $ {\cal R}_{\rm pp} $, $ {\cal R}^{\rm cr}_{\rm ph} $ and $ {\cal R}^{\rm d}_{\rm ph} $ to the particle-particle  and the crossed 
and direct particle-hole terms $ {\cal T}_{\rm pp} $, $ {\cal T}^{\rm cr}_{\rm ph} $ and $ {\cal T}^{\rm d}_{\rm ph} $, respectively.
\begin{widetext}
\begin{align*}
  {\cal R}_{\rm pp} (k_1,k_2,k_3)  = - \sum_{\beta =2,3} \int\! dq \, S(q) \, G_\beta(l-q) & \left[ U_{3,\beta} (l-q,q,k_1) \, U_{3,\beta} (l-q,q,l-k_3) \right. \\
 +  & \left.  U_{3,\beta} (l-q,q,k_2) \, U_{3,\beta} (l-q,q,k_3) \right]_{l=k_1+k_2} \, ,
\end{align*}
\begin{align*}
  {\cal R}^{\rm cr}_{\rm ph} (k_1,k_2,k_3)   = & - \sum_{\beta =2,3} \int\! dq  \, S(q) \, G_\beta(l+q)  \left. U_{3,\beta} (l+q,k_1,q) \, U_{3,\beta} (l+q,k_2-l,q) \right|_{l=k_3-k_1} \\
   & - \sum_{\beta =2,3} \int\! dq  \, S(q) \, G_\beta(l+q)  \left. U_{3,\beta} (l+q,k_2,q) \, U_{3,\beta} (l+q,k_3,q) \right|_{l=k_1-k_3} \, ,
\end{align*}  
\begin{align*} 
  {\cal R}^{\rm d}_{\rm ph} & (k_1,k_2,k_3) = - \sum_{\beta =2,3} \int\! dq \,  S(q) \, G_\beta(l+q)  \left[ -2 U_{3,\beta} (l+q,k_1,k_1+l) U_{3,\beta} (l+q,k_3,k_2) \right. \\
 & +  \left.  U_{3,\beta} (l+q,k_1,k_1+l) \, U_{3,\beta} (l+q,k_3,q) + U_{3,\beta} (l+q,k_1,q) \, U_{3,\beta} (l+q,k_3,k_2) \right]_{l=k_2-k_3} \\
 - & \sum_{\beta =2,3} \int\! dq \,  S(q) \, G_\beta(l+q)  \left[ -2 U_{3,\beta} (l+q,k_2,k_3) \, U_{3,\beta} (l+q,k_1-l,k_1) \right. \\ 
 & +  \left.  U_{3,\beta} (l+q,k_2,k_1-l) \, U_{3,\beta} (l+q,k_3,q) + U_{3,\beta} (l+q,k_2,q) \, U_{3,\beta} (l+q,k_1-l,k_1) \right]_{l=k_3-k_2} \, .
\end{align*} 
\end{widetext}
 In the band-mixed loops,
$ G_\beta (q) = \left[ i q_0 - \epsilon_\beta ({\bf q}) \right]^{-1} $
denotes the propagator of for valence-band hole with band index $ \beta = 2,3 $ and the single-scale propagator $ S(q) $  is given by the scale-derivate of the propagator $ G(q) $ for the conduction band.
Zero-frequency projection allows again for an analytical evaluation of the Matsubara sums. 

\section{Three group-theoretic corollaries on exchange parametrizations} \label{sec:group-th}
 In this appendix, we show that the statements made in Sec.~\ref{sec:FFE} directly follow from the two lemmas named after Schur. (In the literature, as in Ref.~\onlinecite{tinkham-groupth}, Chap.~3-2, the lemma we call Schur's first is often simply referred to as Schur's lemma. The lemma we call Schur's second appears as a nameless lemma right below the first in Ref.~\onlinecite{tinkham-groupth}.) 
 
 In this Appendix, we will not restrict ourselves to a specific point group but discuss an exchange parametrization 
 for some general point group ${\cal G}$. Therefore, the following statements are not limited to a particular lattice geometry or to the symmetric phase.
More precisely, we consider a coupling function $ \Phi (l,p,q) $ that depends strongly on $ l$ and weakly on $p$ and $q$.
This coupling function should be symmetric under the point group ${\cal G}$, i.e.\ we require $ \Phi (R_{\hat{O}}l,R_{\hat{O}}p,R_{\hat{O}}q) = \Phi (l,p,q) \,\, \forall \, \hat{O} \in {\cal G} $,
where the $ R_{\hat{O}} $ are rotation operators acting on the respective momenta.
 If the dependence of $\Phi$ on the weak frequencies $ p_0 $ and $ q_0$ is then dropped, it can then be expanded in form factors  $ f_i $ that transform according to irreducible representations (IRs) of the point-group $ {\cal G} $ of $\Phi $. This expansion reads as
\begin{equation*}
 \Phi (l,p,q) = \sum_{ij} f_i ({\bf p}-{\bf l}/2) \, f_j ({\bf q} \pm {\bf l}/2) \, P_{ij} (l) \, ,
\end{equation*}
 where the sign in the argument of $ f_j$ is $-$ in the particle-particle channel(s) and $+$ in the particle-hole
 channels. Since in this ansatz the dependence on the weak frequencies $p_0 $ and $ q_0$ is suppressed, the form factors can be chosen real.

 In a Hubbard-Stratonovich spirit, $ P_{ij} (l) $ may be interpreted as the propagator of an exchange boson. The $1+D$ momentum $l$ then corresponds to the center-of-mass motion of this composite particle, while $ p- l/2 $ and $ q \pm l/2 $ are the momenta of the relative motion of its constituents --- two electrons or two holes in the particle-particle channel(s) and one electron and one hole in the particle-hole channels. In this picture, the form factors then play the role of fermion-boson vertices with the indices $i$ and $j$ labeling different bosonic flavors.

If the form factors are chosen to be orthonormal, i.e.\ if 
\begin{equation*}
  \int \! d {\bf q} \, f_i ({\bf q}) \, f_j ({\bf q}) = \delta_{ij} \, ,
\end{equation*}
 the bosonic propagator can be straightforwardly obtained from a given $\Phi$ by applying the projection rule
\begin{align} \notag
 P_{ij} (l) & = \! \int \!\! d{\bf p} \, d{\bf q} \, \, f_i({\bf p}-\frac{\bf l}{2})\, f_j ({\bf q} \pm \frac{\bf l}{2}) \\ \label{eqn:proj-rule}
 & \quad \, \times \Phi \left( l,(l_0/2,{\bf p}),(\mp l_0/2,{\bf q}) \right) \, .
\end{align}
 
 We now continue with the proof of our non-mixing conjecture.
 \begin{corollary}[No mixing]
Let $ P (l)$ be a bosonic propagator that has been projected out of a ${\cal G} $-symmetric coupling function $\Phi$ 
 according to Eq.~(\ref{eqn:proj-rule}).
  Suppose that, for fixed $l$, $P(R_{\hat{O}} l) = P (l) \, \, \forall \, \hat{O} \in {\cal K} $, where $ {\cal K} $ is a subgroup of $ {\cal G} $. If the basis set of form factors is then organized in blocks corresponding to IRs of $ {\cal K}$, blocks that mix form factors of inequivalent IRs of ${\cal K} $ must vanish. \label{th:nomix}
\end{corollary}
 \begin{proof}
  Consider the projection rule~(\ref{eqn:proj-rule}) for $ P_{ij} (R_{\hat{O}} l) $ and matrices $ M^\alpha_{\hat{O}}$ of 
the $ \alpha$th IR of ${\cal K} $ which transform the form factors according to
\begin{equation*} 
 f_i (R_{\hat{O}} {\bf k} ) = \sum_{i'} \left( M^\alpha_{\hat{O}} \right)_{ii'} f_{i'} ({\bf k}) \, .
\end{equation*} 
  We further label the IR of $ {\cal K} $ that transforms $ f_i $ with $ \alpha$ and the one transforming $f_j$ with $ \beta$.
After substituting the integration variables ${\bf p} $ and ${\bf q} $ by $ R_{\hat{O}} {\bf p} $ and $ R_{\hat{O}} {\bf q} $, respectively, we exploit the point-group symmetry of $\Phi$. For the block $ P^{\alpha \beta} (l) $ relating the
 $\alpha$th and $\beta$th IR, we then find
 \begin{equation} \label{eqn:proptrans}
 P^{\alpha \beta} (R_{\hat{O}}  l) = \left( M^\alpha_{\hat{O}} \right)^\dagger P^{\alpha \beta} (l) \, M^\beta_{\hat{O}} \, .
 \end{equation}
 Since the left-hand side equals $ P^{\alpha \beta} (l) $ according to the premise, Schur's second lemma applies. 
Therefore, if the $ \alpha$th and $\beta$th IR of ${\cal K} $ cannot be related by an equivalence transformation, the block $ P^{\alpha \beta} $ must vanish. \qed
 \end{proof}
 In Sec.~\ref{sec:FFE}, we have used Corollary~\ref{th:nomix} for ${\cal K} = {\cal L}_{\bf l}$, ${\cal L}_{\bf l} $ being the little group of the bosonic momentum ${\bf l}$. If the premise $P(R_{\hat{O}} l) = P (l) \, \, \forall \, \hat{O} \in {\cal K} $ were satisfied as a consequence of some
approximation, this approximation neglects the mixing of inequivalent IRs of ${\cal K}$.

 For the next two corollaries, the following definition appears useful.
\begin{definition}
 We call a set of form factors \emph{well behaved} under a point-group ${\cal K} $ if its elements transform
according to either identical or inequivalent IRs of ${\cal K}$ consisting of unitary matrices.
\end{definition}
 The point-group behavior of the bosonic propagator is governed by the following law.

\begin{corollary}
 Consider again a bosonic propagator $P(l)$ that has been projected out of a ${\cal G} $-symmetric  coupling function $\Phi$ with form factors which are well behaved under ${\cal G}$. Blocks $P^{\alpha \beta} $ relating equivalent \emph{one-dimensional} IRs $\alpha $ and $\beta$ of ${\cal G}$ then are fully ${\cal G}$-symmetric, i.e.\
$ P^{\alpha \beta} (R_{\hat{O}}l) = P^{\alpha \beta} (l) \,\, \forall \, \hat{O} \in {\cal G}$ for arbitrary $l$. \label{th:1dtrans}
\end{corollary}
\begin{proof}
 We observe that, also in the present case, Eq.~(\ref{eqn:proptrans}) holds, with $ \alpha $ and $\beta$ now labeling IRs of the \emph{full} point group ${\cal G} $. This equation is trivially fulfilled if a block vanishes according to Corollary~\ref{th:nomix}. For non-vanishing blocks, well-behaved form factors give rise to
$ M^\alpha_{\hat{O}} = M^\beta_{\hat{O}} $. If $ \alpha $ and $\beta$ then label one-dimensional IRs of $ {\cal G} $, these matrices are just phase factors, which cancel. \qed
\end{proof}

 If the mixing of inequivalent IRs of ${\cal G}$ is neglected, the remaining one-dimensional irreducible blocks of $P(l)$ are hence ${\cal G}$-symmetric.

 Finally, Schur's  first lemma directly gives rise to the following corollary.
 \begin{corollary}
 Suppose that, for fixed $l$, $P(R_{\hat{O}} l) = P (l) \, \, \forall \, \hat{O} \in {\cal K} $, where $P(l)$ is a bosonic propagator obtained from a $ {\cal G}$-symmetric coupling function and where $ {\cal K} $ is a subgroup of $ {\cal G} $. 
  For a well-behaved set of form factors under $ {\cal K}$, the non-vanishing irreducible blocks $ P^{\alpha \beta} (l) $ of $P(l)$ then are a multiple of a unit matrix,
 with $ \alpha $ and $\beta$ labeling IRs of ${\cal K}$. \label{th:unit-mat}
\end{corollary}
\begin{proof}
 Again, Eq.~(\ref{eqn:proptrans}) holds.
 For well-behaved form factors, the representation matrices $ M^\alpha_{\hat{O}} $ and $ M^\beta_{\hat{O}} $ with $ \hat{O} \in {\cal K} $ are equal to
one another, if $P^{\alpha \beta} $ does not vanish according to Corollary~\ref{th:nomix}. Since $ P (R_{\hat{O}} l) = 
P (l) \,\, \forall \hat{O} \in {\cal K}$, Eq.~(\ref{eqn:proptrans}) simply states that $ M^\alpha_{\hat{O}} $ and 
$ P^{\alpha \beta}$ commute $\forall \hat{O} \in {\cal K} $. According to Schur's first lemma, $ P^{\alpha \beta} $ then must be a multiple of a unit matrix. \qed
\end{proof}
\end{appendix}
\bibliography{biblio}

\end{document}